\documentclass[journal]{IEEEtran}

\usepackage[T1]{fontenc}
\usepackage{etoolbox} 
\usepackage[dvipsnames]{xcolor}
\usepackage[bookmarksnumbered,unicode]{hyperref} 
\usepackage{amsmath,amssymb,amsfonts,amsthm}
\newtheorem{assumption}{Assumption}
  
\usepackage{mathtools} 
\usepackage[inline,shortlabels]{enumitem} 
  \setlistdepth{9}
\usepackage{multirow} 
\usepackage[group-separator={,},per-mode=symbol]{siunitx} 
\usepackage{booktabs} 
\usepackage{tikz} 
\usepackage{tikz-3dplot}
	\usetikzlibrary{automata,positioning,shapes,fit,shadows}
	\usetikzlibrary{calc}
  \usetikzlibrary{svg.path}
	\usetikzlibrary{captl} 
	\usetikzlibrary{extshapes}
	\usetikzlibrary{angles, quotes, 3d, arrows, bending} 
	\usetikzlibrary{decorations.pathreplacing} 
  \usepackage{pgfplots} 
	\pgfplotsset{compat=1.12}
	\usepgfplotslibrary{units}  
	\usepgfplotslibrary{fillbetween} 
	\usepgfplotslibrary{patchplots} 
	\usepgfplotslibrary{colormaps}
	\usepgfplotslibrary{statistics}
  \usetikzlibrary{patterns} 
	\pgfplotsset{compat=newest} 
	\pgfplotsset{plot coordinates/math parser=false} 
	\newlength\figureheight 
	\newlength\figurewidth 
\usepackage{xargs} 

\ifdefined\EnableTodoNotes \else \PassOptionsToPackage{disable}{todonotes} \fi 
\usepackage[colorinlistoftodos,prependcaption,textsize=scriptsize]{todonotes} 
\usepackage{wrapfig}
\usepackage{cutwin} 
\usepackage[ruled,vlined,linesnumbered]{myalgorithm2e}

\usepackage{graphicx}
	\graphicspath{ {figs/} }
\usepackage{upgreek}
\usepackage{relsize} 
\usepackage{fancyhdr} 
\usepackage[yyyymmdd,hhmmss]{datetime}

\usepackage{soul} 
\usepackage[tight,footnotesize]{subfigure}
\usepackage{stackengine}
\usepackage{printlen}
\usepackage[all,cmtip]{xy}
\usepackage[ligature, inference]{semantic}
\usepackage{xifthen}
\usepackage{dirtree} 
\usepackage{fancyvrb} 
\usepackage{listings}
\usepackage{cite}
\usepackage{textcomp}
\usepackage{comment}
\usepackage{orcidlink}
\usepackage{bm} 
\usepackage{pgf}
\usepackage{xfp}
\usepackage{xparse}
\usepackage{dsfont} 
\usepackage{nicefrac} 

\newcommandx{\ME}[2][1=]{\todo[linecolor=blue,backgroundcolor=black!10,bordercolor=blue,#1]{ME: #2}}
\newcommandx{\YS}[2][1=]{\todo[linecolor=blue,backgroundcolor=black!10,bordercolor=blue,#1]{YS: #2}}
\newcommandx{\JJ}[2][1=]{\todo[linecolor=purple,backgroundcolor=purple!10,bordercolor=blue,#1]{JJ: #2}}
\newcommandx{\XX}[2][1=]{\todo[linecolor=blue,backgroundcolor=black!10,bordercolor=blue,#1]{#2}}
\newcommandx{\CF}[2][1=]{\todo[linecolor=red,backgroundcolor=black!10,bordercolor=red,#1]{CF: #2}}

\newcommand{\source}[1]{\relax}

\makeatletter
\define@key{todonotes}{done}[]{%
    \setkeys{todonotes}{color=green!50!black,backgroundcolor=green!10!white,bordercolor=green!50!black,textcolor=black!70,%
		caption={\checkmark}}%
}
\makeatother

\newcommand{\eat}[1]{\relax}

\newcommand{\tabref}		[1]{Table~\ref{#1}}
\newcommand{\figref}		[1]{Fig.~\ref{#1}}
\newcommand{\thmref}		[1]{Theorem~\ref{#1}}
\newcommand{\lemref}		[1]{Lemma~\ref{#1}}
\newcommand{\propref}		[1]{Proposition~\ref{#1}}
\newcommand{\exref}			[1]{Example~\ref{#1}}
\newcommand{\algref}		[1]{Algorithm~\ref{#1}}
\newcommand{\secref}		[1]{Section~\ref{#1}}
\newcommand{\subsecref}		[1]{Section~\ref{#1}}
\newcommand{\subsubsecref}	[1]{Section~\ref{#1}}
\newcommand{\defref}		[1]{Definition~\ref{#1}}
\newcommand{\assumpref}		[1]{Assumption~\ref{#1}}
\newcommand{\remref}		[1]{Remark~\ref{#1}}

\newcommand{\probref}		[1]{Problem~\ref{#1}}
\newcommand{\lineref}		[1]{Line~\ref{#1}}
\newcommand{\linesref}		[2]{Lines~\ref{#1}--\ref{#2}}
\newcommand{\Booleans}{\mathbb{B}}
\newcommand{\Naturals}{\mathbb{N}}
\newcommand{\Reals}{\mathbb{R}}
\newcommand{\R}{\mathbb{R}} 

\newcommand{\Integers}{\mathbb{Z}}

\newcommand{\ie}{i.e., }
\newcommand{\eg}{e.g., }

\DeclareMathOperator{\atantwo}{atan2}

\newcommand{\EmptyWord}{\varepsilon} 
\newcommand{\Uniform}{\mathcal{U}} 

\makeatletter
  \def\underbar#1{\underline{\sbox\tw@{$#1$}\dp\tw@=-0.5pt\box\tw@}}
  \def\uunderbar#1{\underline{\sbox\tw@{$\scalebox{0.9}{$#1$}$}\dp\tw@=1.2pt\box\tw@}}
  \def\overbar#1{\overline{\sbox\tw@{$#1$}\ht\tw@=0.9\ht\tw@\dp\tw@\z@\box\tw@}}
  \def\ooverbar#1{\overline{\sbox\tw@{$\scalebox{0.9}{$#1$}$}\ht\tw@=0.95\ht\tw@\dp\tw@\z@\box\tw@}}
\makeatother

\newtheorem{problem}{Problem}
\newtheorem{remark}{Remark}
\newtheorem{definition}{Definition}
\newtheorem{example}{Example}

\renewcommand{\emptyset}{\varnothing} 

\newcommand{\Vertiport}{\mathcal{V}}

\newcommand{\Observer}{\mathcal{O}}

\newcommand{\Noise}{\eta} 
\newcommand{\TrueNoise}{\eta} 
\newcommand{\NoiseValue}{L}

\newcommand{\BoundedVarAzi}{\mu_{\Azimuth}}

\newcommand{\BoundedVarAct}{\mu_{\text{act}}}

\newcommand{\Instant}{\mathrm{in}} 
\newcommand{\Average}{\mathrm{eq}} 

\newcommand{\InsNoiseThreshold}{{{\overline{\Noise}}_{\Instant}}}
\newcommand{\AvgNoiseThreshold}{{{\overline{\Noise}}_{\Average}}}
\newcommand{\Window}{{\Delta \mathrm{t}}}

\newcommand{\NAZ}{\mathcal{Z}}

\newcommand{\Ordinance}{\phi}

\newcommand{\Horizon}{T}

\newcommand{\True}{\top}
\newcommand{\False}{\bot}

\newcommand{\Speed}{v} 
\newcommand{\Velocity}{v} 

\newcommand{\Rpm}{\rho} 
\newcommand{\RPM}{\rho} 

\newcommand{\PosX}{x}
\newcommand{\PosY}{y}
\newcommand{\PosZ}{z}

\newcommand{\Heading}{\theta} 

\newcommand{\State}{\xi}
\newcommand{\StateSpace}{\Xi}
\newcommand{\InputDomain}{\StateSpace}

\newcommand{\Plan}{\pi}
\newcommand{\Planset}{\Pi}

\newcommand{\MotionModel}{f}

\newcommand{\Start}{\mathrm{o}}
\newcommand{\Target}{{f}}
\newcommand{\Goal}{\mathrm{g}}
\newcommand{\NoiseModel}{\NN}

\newcommand{\NAZmar}{\{\NAZ^{\delta}_j\}_{1}^{\NAZCt}}
\newcommand{\Node}{q}
\newcommand{\Rand}{\mathsf{rnd}} 
\newcommand{\fnSampleState}{\fnSample}
\newcommand{\fnGetKinoDist}{\Fn{getKinoDist}} 
\newcommand{\Near}{\mathsf{near}} 
\newcommand{\New}{\mathsf{new}} 
\newcommand{\fnSteer}{\Fn{steer}} 
\newcommand{\fnSteerPBS}{\Fn{steerPBS}} 
\newcommand{\fnSimulate}{\Fn{simulate}} 
\newcommand{\fnPredictNoise}{\Fn{predictNoise}} 
\newcommand{\fnSlidingAverage}{\Fn{getEqNoise}} 
\newcommand{\fnDetectCollision}{\Fn{detectCollision}} 
\newcommand{\Best}{\mathsf{best}} 

\newcommand{\Vdistance}{h} 
\newcommand{\Hdistance}{r} 
\newcommand{\Azimuth}{\varphi} 

\newcommand{\SectorIndex}{m}
\newcommand{\Hypercube}{\mathcal{H}} 
\newcommand{\Vertices}{\mathit{Vert}} 
\newcommand{\FnFoo}{\mathit{f}} 
\newcommand{\FnBar}{\mathit{g}} 
\newcommand{\CompactSet}{\mathcal{K}}

\newcommand{\Dataset}{\mathcal{D}}
\newcommand{\Test}{\textup{test}}

\newcommand{\Error}{\varepsilon} 
\newcommand{\ErrorBound}{\delta} 
\newcommand{\Term}{I} 
\newcommand{\ConstC}{\mathit{C}} 
\newcommand{\ConstD}{\mathit{D}} 
\newcommand{\CostWeight}{\mathit{W}}

\newcommand{\Queue}{\mathcal{Q}} 
\newcommand{\Front}{\mathtt{front}} 
\newcommand{\Enqueue}{\mathtt{enqueue}} 
\newcommand{\Dequeue}{\mathtt{dequeue}} 

\newcommand{\EvtolCt}{\mathtt{M}} 
\newcommand{\NAZCt}{\mathtt{N}} 
\newcommand{\VertiportCt}{\mathtt{V}}

\newcommand{\IndexSet}{\mathcal{I}} 
\newcommand{\SpeedCt}{\mathtt{N}_\Speed} 
\newcommand{\RpmCt}{\mathtt{N}_\Rpm} 
\newcommand{\VdistanceCt}{\mathtt{N}_\Vdistance} 
\newcommand{\HdistanceCt}{\mathtt{N}_\Hdistance} 
\newcommand{\SectorCt}{\mathtt{N}_\Azimuth} 
\newcommand{\IterationCt}{\mathtt{N}_{\textup{iter}}}
\newcommand{\AttemptCt}{\mathtt{N}_{\textup{atmp}}}

\newcommand{\Tree}{\mathcal{T}}

\newcommand{\Cost}{c}

\DeclareMathOperator*{\argmax}{arg\,max}
\DeclareMathOperator*{\argmin}{arg\,min}

\let\Fn\mathtt

\newcommand{\fnSample}{\Fn{sampleRandomState}}

\newcommand{\fnDist}{\Fn{getDist}}

\newcommand{\fnGetCost}{\Fn{getCost}}

\newcommand{\fnExtractPath}{\Fn{extractPath}}

\newcommand{\NN}{\mathcal{N\!\!N}} 

\newcommand{\Schedule}{\mathcal{S}}

\makeatletter
\renewcommand*\env@matrix[1][*\c@MaxMatrixCols c]{%
\hskip -\arraycolsep
\let\@ifnextchar\new@ifnextchar
\array{#1}}
\makeatother

\SetKw{Break}{break}
\SetKwComment{Comment}{$\triangleright$\ }{}

\SetSideCommentRight
\SetCommentSty{text}

\definecolor{azure(colorwheel)}{rgb}{0.0, 0.5, 1.0}
\definecolor{RoyalBlue}{rgb}{0.0, 0.0, 0.5}
\definecolor{ForestGreen}{rgb}{0.13, 0.55, 0.13}
\definecolor{steelblue}{RGB}{70,130,180}
\definecolor{mplblue}{RGB}{31,119,180}
\definecolor{mplorange}{RGB}{255,127,14}

\definecolor{C0}{HTML}{1F77B4} 
\definecolor{C1}{HTML}{FF7F0E} 
\definecolor{C2}{HTML}{2CA02C} 
\definecolor{C3}{HTML}{D62728} 
\definecolor{C4}{HTML}{9467BD} 
\definecolor{C5}{HTML}{8C564B} 
\definecolor{C6}{HTML}{E377C2} 
\definecolor{C7}{HTML}{7F7F7F} 
\definecolor{C8}{HTML}{BCBD22} 
\definecolor{C9}{HTML}{17BECF} 

\definecolor{viridis1}{RGB}{68,1,84}      
\definecolor{viridis2}{RGB}{59,82,139}    
\definecolor{viridis3}{RGB}{33,144,140}   
\definecolor{viridis4}{RGB}{92,200,99}    
\definecolor{viridis5}{RGB}{253,231,37}   

\definecolor{darkgrey176}{RGB}{176,176,176}

\tikzset{pAction/.style={font={\fontsize{7}{0}\selectfont}, RoyalBlue, above, pos=.5}}
\tikzset{pGuard/.style={font={\fontsize{7}{0}\selectfont}, Red}}
\tikzset{pLabel/.style={font={\fontsize{7}{0}\selectfont}, Black!90}}
\tikzset{pDistr/.style={font={\fontsize{7}{0}\selectfont}, Black!90}}
\tikzset{pAssign/.style={font={\fontsize{7}{0}\selectfont}, ForestGreen}}
\tikzset{pArc/.style={-,shorten <=0pt,shorten >=0pt}}

\pgfplotsset{
	every axis/.append style={
		title style={font=\normalsize\bfseries},
		tick style={line width=0.5pt, black!70},
		tick align=inside,
		tick label style={font=\footnotesize},
		legend style={font=\footnotesize},
		axis line style={line width=0.8pt},
		major grid style={gray!30},
  	minor grid style={gray!15},
		colorbar style={
			title style={font=\normalsize\bfseries},
			label style={font=\small},
			tick label style={font=\footnotesize},
			tick style={line width=0.5pt, black},
			line width=1pt,
			tick align=inside
		}
	}
}

\pgfplotsset{
	PlotStyleA/.style={smooth, very thick, C0, mark=*, mark size=1.5, mark options={solid,fill=C0}}, 
	PlotStyleB/.style={smooth, ultra thick, C0, dash pattern=on 1pt off 1pt, mark=none, mark size=1.5, mark options={solid,fill=C0}}, 
	PlotStyleC/.style={smooth, very thick, C1, mark=square*, mark size=1.5, mark options={solid,fill=C1}}, 
	PlotStyleD/.style={smooth, ultra thick, C1, dash pattern=on 1pt off 1pt, mark=none, mark size=1.5, mark options={solid,fill=C1}}, 
	PlotStyleE/.style={smooth, very thick, C2, mark=triangle*, mark size=2.5, mark options={solid,fill=C2}}, 
	PlotStyleF/.style={smooth, ultra thick, C2, dash pattern=on 1pt off 1pt, mark=none, mark size=2.5, mark options={solid,fill=C2}}, 
	PlotStyleG/.style={ultra thick, black!90!white, dash pattern=on 6pt off 2pt on 2pt off 2pt}, 
	PlotStyleH/.style={ultra thick, black!90!white, dash pattern=on 2pt off 2pt}, 
	PlotStyleI/.style={smooth, very thick, C3, mark=*, mark size=2pt, mark options={solid}}, 
	PlotStyleIB/.style={smooth, very thick, C4, mark=*, mark size=2pt, mark options={solid}}, 
	PlotStyleIC/.style={smooth, very thick, C5, mark=*, mark size=2pt, mark options={solid}}, 
	PlotStyleJ/.style={viridis2, ultra thick, mark=*, mark options={fill=viridis2, solid}, mark size=2pt, smooth}, 
	PlotStyleJA/.style={viridis1, ultra thick, mark=*, mark options={fill=viridis1, solid}, mark size=2pt, smooth}, 
	PlotStyleJB/.style={viridis2, ultra thick, mark=*, mark options={fill=viridis2, solid}, mark size=2pt, smooth, dashed}, 
	PlotStyleJC/.style={viridis3, ultra thick, mark=*, mark options={fill=viridis3, solid}, mark size=2pt, smooth}, 
	PlotStyleProjection/.style={color=gray,dashed,line width=1.5pt,opacity=0.7,mark=*,mark size=1.5pt,mark options={fill=gray, opacity=0.5}, smooth},
	PlotStyleDropLines/.style={color=gray!50,line width=1.0pt,opacity=0.5}
}%

\let\DeclareUSUnit\DeclareSIUnit

\DeclareUSUnit\inch{in}
\DeclareUSUnit\foot{ft}
\DeclareUSUnit\mile{mi}
\DeclareUSUnit\feet{ft}
\DeclareSIUnit{\rpm}{rpm}
\DeclareSIUnit{\mps}{m/s}
\DeclareSIUnit{\decibel}{dB}

\makeatletter
\@ifpackageloaded{siunitx}{%
}{%
  \newcount\ppnum
  \newcommand\num[1]{%
    \ppnum=#1\relax
    \ifnum\ppnum<0
      $-$%
      \ppnum=-\ppnum
    \fi
    \let\pptemp\empty
    \loop\ifnum\ppnum>999
      \count255=\ppnum
      \divide\ppnum by1000
      \count255=\numexpr \count255 - 1000*\ppnum \relax
      \edef\pptemp{,\ifnum\count255<100 0\ifnum\count255<10 0\fi\fi
      \the\count255 \pptemp}%
    \repeat
    \the\ppnum
    \pptemp
  }
}
\makeatother

\newcommand{\AlignedInterval}[2]{%
  $[\text{\makebox[4.5em][r]{\num[round-mode=places,round-precision=2]{#1}$^\circ$}},
    \text{\makebox[4.5em][r]{\num[round-mode=places,round-precision=2]{#2}$^\circ$}})$%
}

\theoremstyle{plain}
\newtheorem{theorem}{Theorem}
\ifdefined\proposition\else
  \newtheorem{proposition}{Proposition}

\newlength{\FigWidth}
\newlength{\FigHeight}
\newlength{\FigXGap}
\newlength{\FigYGap}
\newlength{\YlabelYShift}
\newlength{\TitleYShift}
\newlength{\LineWidth}
\newlength{\MarkSize}

\newcounter{experiment}[section] 

\newcommand{\experiment}[1]{%
  \refstepcounter{experiment}%
  \subsubsection*{Experiment \theexperiment\ --- #1}%
}

\ifdefined\EnableTodoNotes
	\makeatletter
	\renewcommand{\todo}[2][]{\leavevmode\tikzexternaldisable\@todo[#1]{#2}\tikzexternalenable}
	\makeatother
\fi
\usepgfplotslibrary{external}

\ifdefined\EnableTodoNotes
  \paperwidth=\dimexpr \paperwidth + 6cm\relax
  \oddsidemargin=\dimexpr\oddsidemargin + 3cm\relax
  \evensidemargin=\dimexpr\evensidemargin + 3cm\relax
  \marginparwidth=\dimexpr\marginparwidth + 3cm\relax
\else
	\PassOptionsToPackage{disable}{todonotes}
\fi

\setlist[itemize]{leftmargin=14pt} 

\def\BibTeX{{\rm B\kern-.05em{\sc i\kern-.025em b}\kern-.08em
    T\kern-.1667em\lower.7ex\hbox{E}\kern-.125emX}}
    
\begin{document}

\title{%
Certified Learning-Enabled Noise-Aware Motion Planning for Urban Air Mobility
}
\author{
  Jaejeong Park\textsuperscript{1}\orcidlink{0009-0002-1011-8494}, 
  Mahmoud Elfar\textsuperscript{1}\orcidlink{0000-0002-5579-1255}, 
  Cody Fleming\textsuperscript{2}\orcidlink{0000-0001-6335-471X}, and
  Yasser Shoukry\textsuperscript{1}\orcidlink{0000-0002-8224-8477} \\ 
\textit{\textsuperscript{1}University of California, Irvine}, Irvine, CA, USA \\
\textit{\textsuperscript{2}Iowa State University}, Ames, IA, USA
}

\maketitle
\begin{abstract}


Urban Air Mobility (UAM) has emerged as a promising solution to alleviate urban congestion and transportation challenges.
Nevertheless, the noise generated by eVTOL aircrafts poses a significant barrier to public acceptance and regulatory approval, potentially limiting the operational scope and scalability of UAM systems.
Hence, the successful adoption of UAM systems hinges on the ability to predict generated noise levels, and further develop motion planning strategies that comply with community-level noise regulations while maintaining operational efficiency.
To this end,
this paper proposes a novel noise-aware motion planning framework for UAM systems that ensures compliance with noise regulations.
We first develop a certifiable neural network model to accurately predict eVTOL noise propagation patterns in urban environments, providing provable bounds on its correctness.
To achieve a desired level of accuracy, we propose an active sampling strategy to efficiently build the dataset used to train and test the noise model.
Next, we develop a noise-aware motion planning algorithm that utilizes the noise model to generate eVTOL trajectories that guarantee compliance with community noise regulations.
The algorithm exploits the monotonic structure of the noise model to efficiently sample the configuration space, ensuring that the generated trajectories are both noise-compliant and operationally efficient.
We demonstrate the effectiveness of the proposed framework through a number of experiments for Vahana eVTOLs.
The results show that the framework can generate noise-compliant flight plans for a fleet of eVTOLs that adhere to community noise regulations while optimizing operational efficiency.

\end{abstract}

\section{Introduction}
\label{sec:intro}

Urban air mobility (UAM) systems employ electric vertical take-off and landing (eVTOL) aircrafts to provide on-demand transportation of both passengers and goods within urban areas.
In recent years, UAM has gained attention as a promising, transformative technology that could revolutionize urban transportation by offering faster, more efficient, and environmentally friendly travel options.
This can be seen in the growing number of companies investing in eVTOL technology~\cite{dollSouthwestAirlinesSigns2024},
as well as the increasing number of cities exploring the possibility of integrating eVTOLs into their transportation networks~\cite{rimjha2021commuter,qu2024demand}.
Recent Federal Aviation Administration (FAA) draft advisory circulars outline guidelines for the certification of eVTOL aircrafts~\cite{faa2025ac}, further solidifying the potential of UAM as a viable transportation solution. 

A widespread deployment of UAM systems, however, faces several challenges, one of which is the noise pollution generated by eVTOLs.
Despite the use of electric propulsion systems, the noise generated by eVTOLs can have a significant and adverse impact on the quality of life of the general public.
Since UAM systems are expected to operate within and across urban areas to serve a large number of passengers,
airports for vertical take-off and landing, also known as vertiports, are typically located in or near densely populated areas~\cite{garrow2021urban}.
This, and the fact that eVTOLs typically operate at relatively low altitudes, renders the noise generated by eVTOLs a critical concern for the general public,
especially to noise-sensitive locations such as schools, hospitals and retirement homes. 
The adverse impact of UAM noise pollution is compounded for affected communities for which the eVTOL services may not be accessible, raising concerns regarding the equity of UAM deployment.
Consequently, the adoption of noise mitigation strategies is crucial for the successful and equitable deployment of UAM systems, and is expected to be enforced by local and federal regulations.

Nevertheless, noise mitigation in UAM systems comes with its own set of challenges.
The noise generated by eVTOLs is highly nonlinear and dependent on various factors related to
the eVTOL construction (\eg the number of rotors, and their rotational speed and configuration), and
the flight conditions (\eg aircraft altitude, speed and trajectory).
This makes it difficult to accurately predict the noise generated by eVTOLs in real-time,
especially for on-demand UAM services where the flight paths and conditions are constantly changing.
Several studies have approached this challenge using physics-based frameworks that combine detailed aircraft dynamics with aeroacoustic prediction tools. For instance, several studies have approached this challenge using physics-based frameworks that combine detailed aircraft dynamics with aeroacoustic prediction tools. Their framework integrates force-balance kinematic profiles with rotor, airframe, and interaction noise models from NASA’s ANOPP2~\cite{Zorumski1982ANOPP}, enabling assessment of community noise across multiple eVTOL configurations~\cite{yeungFlightProcedureCommunity}~\cite{pelleritoFlightProceduralNoise2024}~\cite{pellerito2024impact}. 
While such physics-based methods can yield highly accurate predictions of acoustic impact under different operational profiles, they are often computationally intensive, which makes them challenging to apply directly in real-time UAM operations. For this reason, there are recent studies employ a supervised U-shaped convolutional neural network with multiple output layers to capture noise propagation in three-dimensional urban environments~\cite{he2025uavnoise}~\cite{kim2024uamn}.

Beyond accurate noise measurement, the large-scale integration of eVTOL operations into urban airspaces requires routing strategies that are both computationally tractable and operationally efficient. 
While UAM-specific noise regulations are still under development and may vary across jurisdictions, they are expected to follow the same notions used in noise ordinances from the FAA and European Union Aviation Safety Agency (EASA) for conventional aircraft operations, including noise abatement zones, noise level limits, and quiet periods~\cite{EASA2021UAM}~\cite{FAA2024_NoisePolicyReview}.
Thus, the noise generated by eVTOLs must be considered during the motion planning phase to ensure compliance with the noise
regulations imposed by the affected communities.
The literature has explored a range of optimization formulations, each tailored to different aspects of UAM networks such as vertiport placement~\cite{wei2024risk}, fairness in resource allocation~\cite{yu2024alpha}~\cite{suIncentiveCompatibleVertiportReservation2024}, and dynamic routing under uncertainty~\cite{wei2023dynamic}.

Finally, recent research has moved beyond pure efficiency and begun to explicitly integrate community noise into trajectory and routing optimization, acknowledging that sustainable UAM operations must balance mobility benefits with urban livability. Representative approaches include a noise-constrained vehicle-routing formulation that couples a virtual flight-noise simulator with a three-phase heuristic to minimize path length and exposure under instantaneous limits~\cite{tan2024lownoise}, and a multilayer, multiobjective optimization that allocates UAM traffic while controlling and fairly distributing community noise via convex–concave procedures~\cite{gao2024noise}.

In this paper, we propose a framework for noise-aware motion planning in UAM systems that conforms to the noise regulations imposed by the affected communities.
Specifically, our contributions are threefold.
First, we develop a neural network (NN) noise model capable of rapidly predicting eVTOL noise levels with provable accuracy bounds.
A high-fidelity dataset of tonal and broadband noise levels under various flight conditions is generated to train the model, ensuring computational efficiency for real-time applications.
To guarantee correctness, we employ monotonic NNs~\cite{nolte2023expressive} to exploit the physical property of acoustic noise of eVTOL.
Moreover, we design an active sampling strategy to efficiently build the dataset used to train and test a noise model so that it is certified to adhere to a given accuracy threshold.

Second, we design a noise-aware motion planning algorithm to generate eVTOL trajectories that guarantee compliance with community noise regulations.
The algorithm modifies the RRT* algorithm by incorporating the community noise ordinances as constraints in the planning process.
It also exploits the monotonic structure of the noise model to efficiently sample the configuration space, improving the computational efficiency of the planner.

Finally, we show the experimental evaluation of the proposed framework through several experiments in simulated UAM scenarios.
We first demonstrate the correctness of the monotonic NN-based noise model that we developed for the Airbus Vahana eVTOL aircraft.
We then evaluate the noise-aware motion planning algorithm in a case study of an on-demand UAM system.
The experiments examine the impact of noise regulation levels and different air traffic density on the generated flight plans.
The results show that the noise-aware motion planner can generate flight paths that comply with community noise regulations while optimizing operational efficiency.

In summary, the main contributions of this paper are as follows:
\begin{itemize}
	\item
	      We develop a monotonic NN-based noise model trained on high-fidelity data to predict eVTOL noise propagation in real time.
	      The model is designed for computational efficiency and is equipped with provable bounds on its accuracy, ensuring reliable noise estimation.
	\item 
            We introduce active sampling method that selectively augments the dataset only when a verifiable criterion indicates that additional measurements will reduce the provable bounds. This procedure tightens the provable accuracy bound of the monotonic NN-based noise model while avoiding the inefficiency of uniform sampling, which even collects unnecessary data in regions that do not contribute to tightening the bound.
	\item
	      We propose a noise-aware motion planning framework that integrates noise constraints into a modified RRT*-based planner.
	      The planner operates in an augmented configuration space, ensuring that generated flight paths comply with community noise regulations while optimizing operational efficiency.
          
        \item
        We demonstrate the effectiveness of the noise-aware motion planner in generating compliant flight paths to noise ordinances through two case studies: (i) a single-eVTOL study that tests the planner under increasingly strict ordinance scenarios, and (ii) a multi-eVTOL study that launches eVTOLs from different vertiports and efficiently mediates cumulative noise by allocating and tracking each eVTOL’s noise budget during trajectory planning.
\end{itemize}

\section{Problem Formulation}
\label{sec:problem}

\subsubsection*{Notation}

We use $\Naturals$, $\Integers$, $\Reals$, $\Reals_{\geq 0}$  to denote the sets of natural, integer, real numbers, and nonnegative real numbers, respectively, and $\varnothing$ to denote the empty set.
For a bounded variable $x$, we use $\underbar{x}$ and $\overbar{x}$ to denote its lower and upper bounds, respectively, \ie $x \in [\underbar{x}, \overbar{x}] \subset \Reals$.
For intervals over integers, we use $[a:b]$ to denote $\{a, a+1, \ldots, b-1, b\}$,
where $a, b \in \Integers$ and $a \leq b$.
For a sequence or a tuple of variables $A$ of size $n$, we use $A[i]$ to denote the $i$-th element of $A$,
where $i \in [1:n]$ and $n \in \Naturals$.
We use $\mathds{1}_{x \in X}$ to denote the indicator function that is 1 if $x \in X$ and 0 otherwise.
Given a tuple $T=(A,B,C)$, we use the dot $.$ notation to denote the elements inside the tuple $T$, i.e., we use $T.A$, $T.B$, and $T.C$ for the $A$, $B$, and $C$ components of $T$.

\subsection{Noise-aware UAM Planning Problem}

Consider an on-demand UAM system that serves a number of communities as shown in~\figref{fig:evtol_example}.
We define the state of an eVTOL as a tuple:
\begin{align}
\State =
  \left(\Velocity, \Rpm, \PosX,
  \PosY, \PosZ, \Heading\right),\;
  \State \in \StateSpace,
  \label{eq:evtol_state}
\end{align}
where:
\begin{itemize}
  \item $\Velocity \in \Reals_{\geq 0}$ is the eVTOL's cruising speed,
  \item $\Rpm \in \Reals_{\geq 0}$ is its rotor speed,
  \item $(\PosX, \PosY, \PosZ) \in \Reals^3$ is its position with respect to a global coordinate system,
  \item $\Heading \in \Reals$ is its heading angle measured from the $x$-axis in the horizontal plane, and
  \item $\StateSpace \subset \Reals^6$ is the eVTOL state space.
\end{itemize}

A centralized \emph{Air Traffic Management System} (ATMS) is responsible for coordinating eVTOL flights to ensure safe and efficient operations.
For an on-demand UAM system, we assume that the ATMS generates flight schedules dynamically based on a first-come, first-serve approach to incoming flight requests rather than following fixed timetables.
Once a flight request is received, the ATMS is responsible for generating a motion plan for the eVTOL--understood here as the sequence of full eVTOL states defined in~\eqref{eq:evtol_state}, not merely a geometric $(x,y,z)$ path--while ensuring safe and efficient navigation and avoiding conflicts with already-approved itineraries.

\begin{definition}[eVTOL Motion Plan]
  \label{def:motion_plan}
  A motion plan for an eVTOL over a time horizon $\Horizon \in \Naturals$
  is a sequence
  \begin{align*}
    \Plan^\Horizon = \State^{(0)} \State^{(1)} \ldots\, \State^{(\Horizon)} \; \subseteq \StateSpace^{\Horizon+1} \;\text{s.t.} \quad \State^{(t+1)} = \MotionModel (\State^{(t)})\,, 
  \end{align*}
  where
  $\State^{(t)}$ denotes the eVTOL state at time $t\in [1:\Horizon-1]$, and 
  $\MotionModel : \StateSpace \to \StateSpace$ represents the eVTOL motion model.
  We denote by $\Planset^\Horizon$ the set of all possible $\Horizon$-length motion plans for the eVTOL.
\end{definition}

To monitor the noise generated by airborne eVTOLs, we define an \emph{observer} as a fixed location at $\Observer = (\PosX_\Observer, \PosY_\Observer, \PosZ_\Observer)$ where noise levels are measured.
We define the state of the eVTOL relative to the observer as a tuple:
\begin{align}
\State_{\Observer} = (\Speed, \Rpm, \Vdistance, \Hdistance, \Azimuth),\;
\State_{\Observer} \in \InputDomain_{\Observer},
\label{eq:relative_state}
\end{align}
where:
\begin{itemize}
	\item $\Vdistance = \PosZ - \PosZ_{\Observer}$ is the vertical distance to the observer,
	\item $\Hdistance = \sqrt{(\PosX - \PosX_{\Observer})^2 + (\PosY - \PosY_{\Observer})^2}$ is the horizontal distance to the observer,
	\item $\Azimuth = \theta - \atantwo\left(\PosY - \PosY_{\Observer},\, \PosX - \PosX_{\Observer}\right)$ is the azimuthal angle from the observer to the eVTOL, and
	\item $\InputDomain_{\Observer} \subset \R^5$ is the state space relative to the observer.
\end{itemize}
A \emph{noise function} $\TrueNoise : \InputDomain_{\Observer} \to \Reals_{\geq 0}$ captures the noise level generated by the eVTOL at observer $\Observer$ as a function of its relative state $\State_{\Observer}$.
Many existing noise ordinances for aircrafts (\eg~\cite{SanDiegoNoiseOrdinance}) regulate noise levels by defining maximum noise thresholds that vary by community.
Specifically, we consider two noise metrics because, together, they provide complementary constraints—protecting against both brief loud and sustained exposure events—commonly required by noise ordinances~\cite{rizzi}~\cite{begault2018uammtrcs}:
\begin{itemize}
  \item \emph{the instantaneous noise level} $\Noise(\State_{\Observer}^{(t)})$ at observer $\Observer$ and time $t$, and
  \item \emph{the equivalent continuous noise level} $\Noise_{\Average}(\State_{\Observer}^{(t)})$ at observer $\Observer$ over a fixed time window $\Window \in \Naturals$:

  \begin{align*}
  \Noise_{\Average}(\State_{\Observer}^{(t)})
= 10 \log_{10}\!\left(
   \frac{1}{\Window} \sum_{j = t - \Window}^{t} 10^{\,\Noise(\State_\Observer^{(j)})/10}
  \right).
  \end{align*}
  
\end{itemize}


We formalize the notion of noise abatement zones as follows.

\begin{definition}[Noise Abatement Zone]
  \label{def:naz}
  A noise abatement zone (NAZ) is a tuple:
	$\NAZ_n = \left(\Observer, \InsNoiseThreshold, \AvgNoiseThreshold, {\Window}, \Ordinance\right)$,
  where:
  \begin{itemize}
		\item $\Observer \in \Reals^3$ is the observer $\Observer$ location at which noise is regulated;
    \item $\InsNoiseThreshold \in \Reals_{\geq 0}$ is the instantaneous noise threshold (in dBA)\footnote{dBA denotes A-weighted sound pressure level, which applies a standardized perceptual weighting; it is the convention for environmental and community noise reporting, unlike unweighted dB.} that must not be exceeded at observer $\Observer$ of $\NAZ$;
    \item $\AvgNoiseThreshold \in \Reals_{\geq 0}$ is the equivalent continuous noise threshold (in dBA) that must not be exceeded at observer $\Observer$ over a fixed time period;
    \item $\Window \in \Naturals$ is the time window over which the average noise $\AvgNoiseThreshold$ is computed; and
		\item $\Ordinance \in \Booleans$
			is a predicate that captures the noise ordinance satisfaction conditions at observer $\Observer$ based on the noise function $\Noise$, and is satisfied at time $t$ if the following condition hold:
            \begin{align}
            \Ordinance (\Noise(\State_{\Observer}^{(t)})) \!=\!
            \begin{cases}
            \True, \; \text{if } \Bigl(\!\Noise(\State_{\Observer}^{(t)}) \leq \InsNoiseThreshold \Bigl) \text{\,and\,} 
            					\Bigl(\!\Noise_{\Average}(\State_{\Observer}^{(t)}) \leq \AvgNoiseThreshold\Bigl),\\
            \False, \; \text{otherwise.}
            \label{eq:ordinance_conditions}
            \end{cases}
            \end{align}
    \end{itemize}

\end{definition}

\begin{figure}[t] 
	\centering
	\pgfplotsset{
		every axis/.append style={
			grid=none,
			title style={yshift=-0.5ex},
			xlabel style={yshift=0.5ex},
			}
		}
	\setlength{\FigWidth}{5in}\setlength{\FigHeight}{2.6in}
	\setlength{\FigXGap}{0.5in}\setlength{\FigYGap}{0.2in}
	\setlength{\YlabelYShift}{-5pt}\setlength{\TitleYShift}{-10pt}
	\setlength{\LineWidth}{1.6pt}\setlength{\MarkSize}{2.6pt}
	
    \includegraphics[width=0.98\columnwidth]{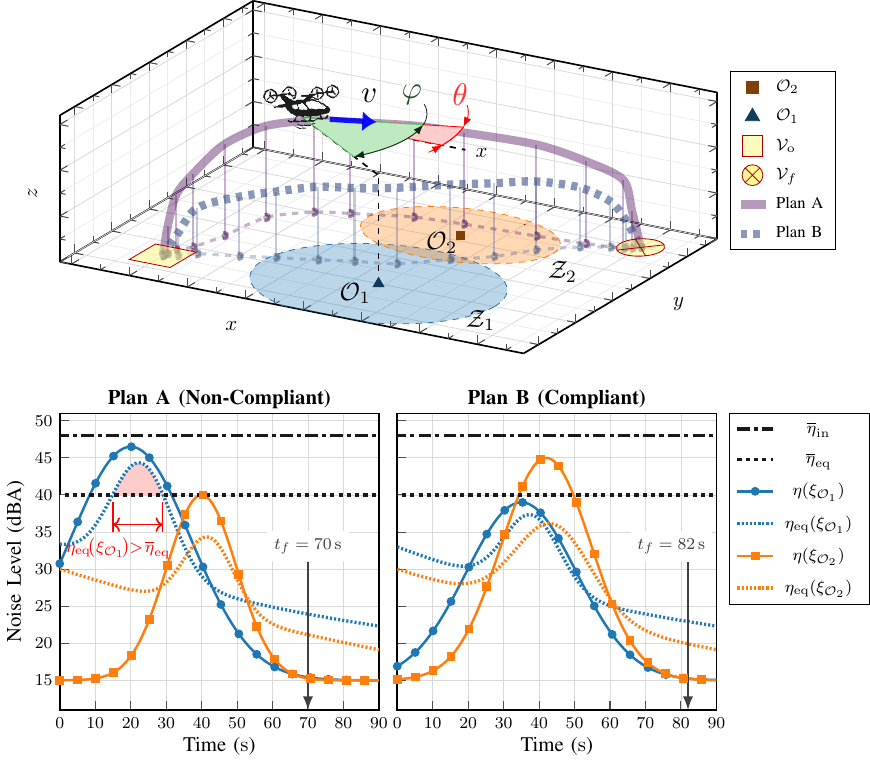}
	\caption{%
		Illustration of an eVTOL at state $\State = (\Speed, \Rpm, \PosX, \PosY, \PosZ, \Heading)$ traveling from vertiport $\Vertiport_{\Start}$ to $\Vertiport_{\Target}$ (top); and
		the instantaneous and equivalent continuous noise levels at observers $\Observer_1$ and $\Observer_2$ for two candidate motion plans A and B (bottom).
		The interval in which the equivalent continuous noise level $\Noise_{\Average}(\State_{\Observer_1})$ exceeds the threshold is highlighted in red.
		The eVTOL's state with respect to $\Observer_1$ is shown in the top panel.
	}
		\label{fig:evtol_example}
\end{figure}

\begin{example}
\figref{fig:evtol_example} 
shows a noise-regulated airspace with two noise-abatement zones: $\NAZ_{1}$ (observer $\Observer_1$) and $\NAZ_{2}$ (observer $\Observer_2$). For this example, $\NAZ_{1}$ and $\NAZ_{2}$ use the same $\AvgNoiseThreshold$ and $\InsNoiseThreshold$ thresholds; in general, thresholds may differ across NAZs. The ATMS plans a flight from $\Vertiport_{\Start}$ to $\Vertiport_{\Target}$. Two candidate plans, A and B, are evaluated.\\
\textbf{Plan A.} It minimizes travel time by flying low through $\NAZ_{1}$, reaching $\Vertiport_{\Target}$ at $t_{\Target}=\SI{70}{\second}$. The equivalent-continuous noise level $\Noise_{\Average}(\State_{\Observer_{1}})$ exceeds its threshold $\AvgNoiseThreshold$, so Plan A is non-compliant at $\NAZ_{1}$. \\
\textbf{Plan B.} It flies higher and slightly longer to limit exposure near $\NAZ_{1}$, arriving at $t_{\Target}=\SI{82}{\second}$. Plan B satisfies the noise ordinance at both NAZs.
\end{example}

Since the noise generated by an eVTOL is a function of its state,
the ATMS must consider the impact of this noise on, and their proximity to, NAZs when generating flight plans.
To generate a motion plan that complies with the noise ordinances specified by NAZs,
the ATMS must ensure that the noise generated by the eVTOL at each observer in a NAZ does not exceed the noise thresholds specified by the NAZ's ordinance,
which in turn requires the ATMS to predict the noise generated by the eVTOL at the NAZ's observer, given its motion plan.
Hence, we formulate the noise-aware UAM motion planning problem as follows.

\begin{problem}[Noise-Aware UAM Planning]
  \label{prob:motion_planning}
  Consider a UAM system that consists of:
  \begin{itemize}
    \item a set of NAZs $\{ \NAZ_1,\ldots, \NAZ_\NAZCt \}$,
    where
		\begin{align*}
      \NAZ_j = \left(\Observer_j, \InsNoiseThreshold_j, \AvgNoiseThreshold_j, \Window_j, \Ordinance_j \right); 
		\end{align*}
    \item a set of vertiports $\{ \Vertiport_1, \ldots, \Vertiport_\VertiportCt \}$; and
    \item a set of eVTOLs $\{ 1, \ldots, \EvtolCt \}$.
  \end{itemize}
  Also, let
  $\Schedule^\Horizon = \{ \Plan_{1}^{\Horizon}, \ldots, \Plan_{\EvtolCt}^{\Horizon} \} \in \prod_{i=0}^{\EvtolCt} \Planset_{i}^\Horizon $
  be the ATMS's current motion plans for all eVTOLs and time horizon $\Horizon$.
  Given a new flight request from
	vertiport $\Vertiport_{\Start}$ at time $t_\Start$ to $\Vertiport_{\Target}$ at time $t_\Target$
	for some eVTOL $ i \in [1:\EvtolCt]$,
  the problem is to find a motion plan $\Plan_{i} \in \Planset_{i}$ for eVTOL ${i}$ such that:
  \begin{itemize}
    \item $\left(\PosX_{i}^{(t_\Start)},\PosY_{i}^{(t_\Start)},\PosZ_{i}^{(t_\Start)}\right) = \Vertiport_{\Start}$;
    \item $\left(\PosX_{i}^{(t_\Target)},\PosY_{i}^{(t_\Target)},\PosZ_{i}^{(t_\Target)}\right) = \Vertiport_{\Target}$; and
    \item $\forall j \in [1:\NAZCt],\,\forall t \in [t_\Start: t_{\Target}]:$
		$
		\Ordinance_{{j}}(\Noise(\State_{\Observer_{j}}^{(t)})) = \True,
		$
		where $\Noise$ is the noise function measured at the state of the eVTOL relative to the observer $\Observer_{j}$.
  \end{itemize}
\end{problem}

The challenge in solving~\probref{prob:motion_planning} is bifold.
First, since the noise measured by an observer is impacted by the noise generated by nearby eVTOLs,
real-time, noise-aware motion planning requires accurate and efficient noise prediction models.
Second, the motion planning algorithm must utilize both the noise prediction and noise ordinance conditions to formulate the noise-based constraints for the motion planning problem.
Moreover, those noise-based constraints are time-dependent,
which adds another layer of complexity to the planning problem.
\section{Certified Learning-Enabled Noise Prediction}
\label{sec:noise}

In this section, we address the first part of~\probref{prob:motion_planning} by focusing on predicting the noise generated by an eVTOL at state $\State$ measured by a fixed ground observer $ \Observer = (\PosX_\Observer, \PosY_\Observer, \PosZ_\Observer) \in \Reals^3$ as shown in~\figref{fig:evtol_example}.

The main challenge stems from the fact that the noise function $\TrueNoise$ is highly nonlinear due to acoustic waves arising from unsteady aerodynamics governed by complex partial differential equations~\cite{salgueiro2024operational}.
Aeroacoustic noise prediction requires solving the full unsteady aerodynamics of the system,
typically done using high-fidelity flight dynamics solvers and acoustic noise prediction tools.
Commonly used flight simulation tools include FUN3D~\cite{fun3d} and FLOWUnsteady~\cite{flow_unsteady}, while widely adopted acoustic noise prediction tools include ANOPP2~\cite{anopp2} and PSU-WOPWOP~\cite{psu_wopwop}.
These tools enable accurate simulation but are computationally expensive. For instance, running a single simulation using FLOWUnsteady for the Airbus Vahana eVTOL model can require between 1--4 days of compute time, depending on fidelity settings.~\footnote{
In our experiments, collecting the simulation dataset required approximately 2 days of wall-clock time on UCI's High Performance Community Computing Cluster (HPC3)~\cite{uci_rcic_hpc3_specs}.}
Subsequent acoustic post-processing using PSU-WOPWOP significantly increases the total computational time, often requiring additional several hours depending on the spatial resolution and extent of the observer grid.
The computational cost of these high-fidelity tools makes them impractical for online planning, yet planners require fast and reliable noise estimates at observer locations. A neural network (NN) model offers the required computing speed and can incorporate physics-motivated structure, but it does not provide guarantees on worst-case error. Such errors may lead to violations of noise ordinances when the eVTOL operates in untested states. We therefore pose a certified learning problem: to learn an NN-based noise model over a bounded operating domain while providing an explicit bound on its deviation from the reference noise function. The formulation below states this problem precisely.

To that end, we introduce a principled framework for training a neural network model for
accurate and efficient noise prediction.
Our framework exploits the physical properties of sound propagation to derive a physics-informed NN architecture
and a training procedure optimized for modeling the noise levels generated by an eVTOL aircraft at a fixed observer location.
Unlike traditional machine learning approaches that rely solely on data-driven techniques to train a black-box model,
we leverage the underlying physics of noise generation to guide the selection of the NN architecture, activation functions, and training data sampling strategy,
allowing us to derive a certified worst-case bound on the prediction error for the trained model.
\subsection{Certified Learning-Enabled Noise Prediction Problem}

To enable real-time and reliable noise-aware motion planning, we propose a reduced-order approximation of $\TrueNoise$
using a NN-based noise model $\NN^{\textup{noise}}$.
This approximation allows for efficient noise level estimation across varying flight conditions.
However, NNs do not inherently provide guarantees on the accuracy of their predictions.
It is therefore essential to provide certifiable guarantees that the noise levels predicted by the NN are within a user-specified error tolerance.
To this end, we define the first subproblem as follows.

\begin{problem}[Certified Learning-Enabled Noise Modeling]
  \label{prob:noisePrediction}
	Given a bounded input domain $\InputDomain_\Observer$ for observer $\Observer$, and
  a black box noise function
	$\TrueNoise \colon \InputDomain_\Observer \rightarrow \Reals_{\geq 0}$ for an eVTOL,
	train a NN-based noise model
	$\NN^{\textup{noise}}\colon \InputDomain_\Observer \rightarrow \Reals_{\geq 0}$ such that:
  \begin{align}
    \max_{ \State_\Observer \in \InputDomain_\Observer}
    | \TrueNoise(\State_\Observer) - \NN^{\textup{noise}}(\State_\Observer) | \leq \ErrorBound, \nonumber
  \end{align}
  where $\State_\Observer \in \InputDomain_\Observer$ is the eVTOL state relative to the location at observer $\Observer$, and
	$\ErrorBound > 0$ is a user-specified maximum error tolerance.
\end{problem}

\begin{remark}
\label{remark:certified_unseen}
In machine learning, conventional test and validation techniques evaluate a model's performance on a specific subset of data, providing an empirical estimate of error. However, these estimates are limited to the particular data samples used and do not guarantee performance on unseen data. In contrast, the certified learning  in Problem~\ref{prob:noisePrediction} aims to provide a certified error bound, which is a provable worst-case upper bound on the NN-based noise model's error that holds true across the entire input domain $\InputDomain$, not just the training or validation sets. This distinction is crucial because it offers a formal guarantee of robustness, ensuring that even for inputs not encountered during training, the NN-based noise model's error will not exceed a predefined threshold, thereby addressing the limitations of empirical evaluation and hence ensures adherence to noise regulations.
\end{remark}

\subsection{Physics-Informed NN Design}

We begin by formalizing the fundamental physical properties of the noise field of an eVTOL, also referred to as the source, and its impact on an observer positioned at a distance from the aircraft. 
In aeroacoustics, it is well established that overall sound pressure level (OASPL) increases monotonically with both the mechanical energy input of the source (\ie the eVTOL) and its proximity to the observer~\cite{rizzi}.
Conversely, OASPL decreases as the vertical or horizontal distance from the source increases due to geometric spreading and atmospheric attenuation~\cite{hur2023farfield}.
Hence, a proper noise model should reflect these monotonic relationships.

To this end,
we define a monotonicity property for $\TrueNoise$, which constrains its partial derivatives with respect to the eVTOL and observer parameters.

\begin{assumption}[Noise Monotonicity]
\label{assump:noise_monotonicity}
The noise function $\TrueNoise(\Speed, \Rpm, \Vdistance, \Hdistance, \Azimuth)$ satisfies the following constraints:
\begin{align*}
	\frac{\partial \TrueNoise}{\partial \Speed} \ge 0, &\quad
	\frac{\partial \TrueNoise}{\partial \Rpm} \ge 0, \quad
	\text{(source energy monotonicity)} \\
	\frac{\partial \TrueNoise}{\partial \Vdistance} \le 0, &\quad
	\frac{\partial \TrueNoise}{\partial \Hdistance} \le 0. \quad
	\text{(observer proximity monotonicity)}
\end{align*}
\end{assumption}
As stated in~\assumpref{assump:noise_monotonicity}, unlike $v,\rho,h,$ and $r$, the azimuthal angle $\varphi$ does not admit a monotonic relationship with the noise level. Accordingly, we first examine how variations in $\varphi$ affect noise before training the noise model.
Noise sources, such as eVTOLs, emit acoustic energy unevenly in space,
forming directional patterns that are known as \emph{lobes}~\cite{olaveson2022spatiospectral}.
Their exact shapes depend on the structure and motion of the source,
\eg rotating blades or distributed propellers.
Within each lobe, OASPL gradually varies smoothly as a function of the direction of observation (\ie the azimuthal angle $\Azimuth$).
This gradual variation arises from the interference and propagation of acoustic waves in the far field, and has been consistently observed across a range of acoustic systems.
Consequently, the noise level can be modeled as a continuous function of the azimuthal angle, $\Azimuth$, within any sufficiently small angular sector.
This property supports the use of a bounded, direction-dependent noise model over $\Azimuth$.

\begin{assumption}[Bounded Azimuthal Variation Property]\label{assump:angular_monotonicity_noise}
Given an eVTOL as a noise source, and an observer $\Observer$,
the \emph{bounded variation property} of the noise function
$\TrueNoise(\Speed,\Rpm,\Vdistance,\Hdistance,\Azimuth)$,
w.r.t. the azimuthal angle $\Azimuth$ and tolerance $\BoundedVarAzi>0$, is defined as follows:
\begin{align*}
\forall \Azimuth \in &[-\underbar{\Azimuth}, \overbar{\Azimuth}),\; \exists \Azimuth_{\text{crit}} > 0 \;
\;\text{s.t.}\; \forall \Azimuth' \in [-\Azimuth_{\text{crit}}, \Azimuth_{\text{crit}}),\;\\
&\qquad\qquad\big|\TrueNoise(\State_\Observer \mid \Azimuth + \Azimuth') -\TrueNoise(\State_\Observer \mid \Azimuth)\big|\leq \BoundedVarAzi~\textup{dB}.
\end{align*}
\end{assumption}

\begin{figure}[!t]
    \centering
    \includegraphics[width=0.95\columnwidth]{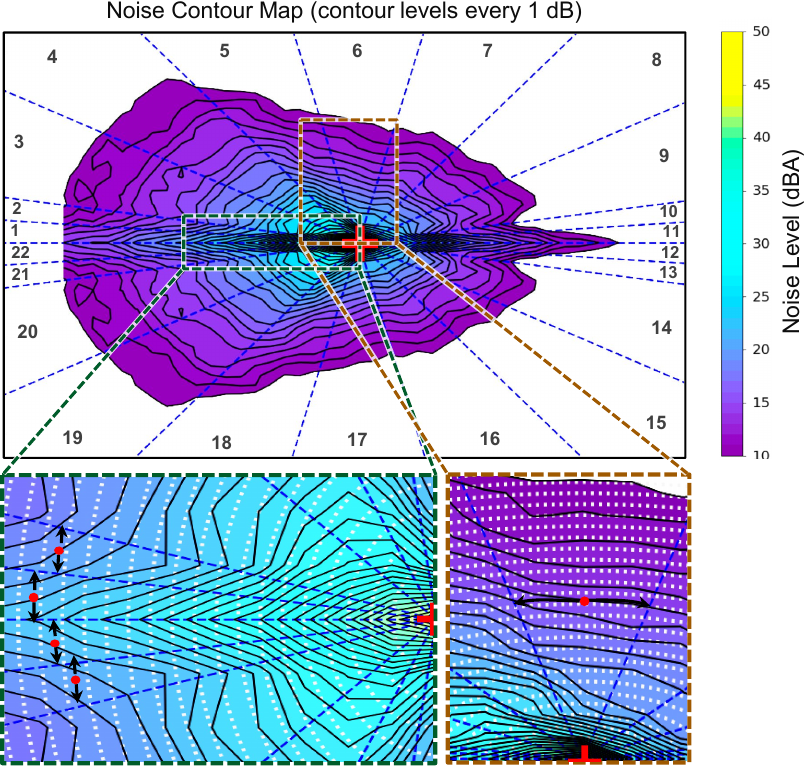}
    \caption{Segmentation of the azimuthal domain into $22$ angular sectors based on the bounded azimuthal variation property ($\BoundedVarAzi = 1$ dB).
		The noise contour map is generated at flight conditions
		$\Speed = \SI{60}{\mps}$, $\Rpm = \SI{700}{\rpm}$, $\Vdistance = \SI{50}{\meter}$.
		}%
    \label{fig:region_subdivision}
\end{figure}

\assumpref{assump:angular_monotonicity_noise} implies that
for any fixed flight conditions $\Speed, \Rpm, \Vdistance, \Hdistance$,
there exists an interval around each azimuthal angle $\Azimuth$ within which
the noise function $\TrueNoise(\Speed, \Rpm, \Vdistance, \Hdistance, \Azimuth)$
deviates by no more than $\BoundedVarAzi$~dB from its value at $\Azimuth$.
That is, the noise levels observed at azimuthal angles within this interval differ by no more than $\BoundedVarAzi$ dB. While $\BoundedVarAzi$ can be chosen freely as a design parameter, in this work we deliberately set $\BoundedVarAzi = 1$ dB, since 1 dB is commonly regarded as the threshold of human hearing perception~\cite{bradley1999just},
and are therefore practically indistinguishable.
This bounded variation property motivates the partition of the full azimuthal domain
into a sequence of angular sectors as follows.
Starting from some initial angle $\underbar{\Azimuth}_1 \in [\underbar{\Azimuth}, \overbar{\Azimuth}]$,
the noise function
$\TrueNoise(\overbar{\Speed}, \overbar{\Rpm}, \underbar{\Vdistance}, \Hdistance,
\underbar{\Azimuth}_1 + \Delta\Azimuth)$ is evaluated at each small angular increment of $\Delta\Azimuth$, until:
\begin{align}
\label{eq:angular_partitioning_condition}
	\big|\TrueNoise(\overbar{\Speed}, \overbar{\Rpm}, \underbar{\Vdistance}, \Hdistance, \underbar{\Azimuth}_1 + \Delta\Azimuth) - \TrueNoise(\overbar{\Speed}, \overbar{\Rpm}, \underbar{\Vdistance}, \Hdistance, \underbar{\Azimuth}_1)\big| > \BoundedVarAzi~\text{dB},
\end{align}
defining the end of the first angular sector
  $[\underbar{\Azimuth}_1, \underbar{\Azimuth}_1 + \Delta\Azimuth)$,
and the start of the next sector
  $\underbar{\Azimuth}_2 = \underbar{\Azimuth}_1 + \Delta\Azimuth$.
By repeating this partitioning process, the azimuthal domain is partitioned into
a finite set of non-overlapping angular sectors:
\begin{align*}
	\left\{ \Azimuth_{\SectorIndex} \,\bigg|\,
		\Azimuth_{\SectorIndex} = \left[\underbar{\Azimuth}_{\SectorIndex}, \overbar{\Azimuth}_{\SectorIndex}\right), \;
		\bigcup_{\SectorIndex = 1}^{\SectorCt} \Azimuth_{\SectorIndex} = [\underbar{\Azimuth}, \overbar{\Azimuth}), \;
		\bigcap_{\SectorIndex = 1}^{\SectorCt} \Azimuth_{\SectorIndex} = \emptyset		
		\right\},
\end{align*}%
where $\SectorCt$ denotes the total number angular sectors that partition the azimuthal domain.
Note that the noise levels in~\eqref{eq:angular_partitioning_condition} are evaluated under the worst-case flight conditions, \ie $\overbar{\Speed}$, $\overbar{\Rpm}$, and $\underbar{\Vdistance}$, to guarantee the validity of the azimuthal decomposition under all flight conditions.
The result of applying this process to the data collected from the noise simulation
is summarized in~\tabref{tab:angular_partitions},
and the resulting angular sectors $\{\Azimuth_{1}, \ldots, \Azimuth_{{22}}\}$ are visualized in~\figref{fig:region_subdivision}.
Further details of the azimuthal partitioning process are provided later in~\secref{sec:experiments}. From this process, each noise function is now fully monotonic in each partition.

\begin{table}[t]
\centering
\caption{Azimuthal partitioning based on bounded noise variation. \\($\BoundedVarAzi$ = 1 \textup{dB})}
\label{tab:angular_partitions}
\resizebox{0.85\columnwidth}{!}{%
\begin{tabular}{cc|cc}
\toprule
$\SectorIndex$ & Angular Range $\left[\underbar{\Azimuth}_{\SectorIndex}, \overbar{\Azimuth}_{\SectorIndex}\right)$ &
$\SectorIndex$ & Angular Range $\left[\underbar{\Azimuth}_{\SectorIndex}, \overbar{\Azimuth}_{\SectorIndex}\right)$ \\
\midrule
1  & $ \AlignedInterval{000.00}{005.62} $    & 22  & $\AlignedInterval{-5.625}{0.00} $  \\
2  & $ \AlignedInterval{005.62}{011.25} $    & 21  & $\AlignedInterval{-11.25}{-5.625} $  \\
3  & $ \AlignedInterval{011.25}{033.75} $    & 20  & $\AlignedInterval{-33.75}{-11.25} $  \\
4  & $ \AlignedInterval{033.75}{056.25} $    & 19  & $\AlignedInterval{-56.25}{-33.75} $  \\
5  & $ \AlignedInterval{056.25}{078.75} $    & 18  & $\AlignedInterval{-78.75}{-56.25} $  \\
6  & $ \AlignedInterval{078.75}{101.25} $    & 17  & $\AlignedInterval{-101.25}{-78.75} $  \\
7  & $ \AlignedInterval{101.25}{123.75} $    & 16  & $\AlignedInterval{-123.75}{-101.25} $  \\
8  & $ \AlignedInterval{123.75}{146.25} $    & 15  & $\AlignedInterval{-146.25}{-123.75} $  \\
9  & $ \AlignedInterval{146.25}{168.75} $    & 14  & $\AlignedInterval{-168.75}{-146.25} $  \\
10 & $ \AlignedInterval{168.75}{174.375} $   & 13  & $\AlignedInterval{-174.375}{-168.75} $  \\
11 & $ \AlignedInterval{174.375}{-180.00} $  & 12  & $\AlignedInterval{-180.00}{-174.375} $  \\
\bottomrule
\end{tabular}%
}
\end{table}

Given the azimuthal partitioning $\{\Azimuth_{1}, \ldots, \Azimuth_{{\SectorCt}}\}$,
we denote by $\InputDomain_{\Observer,\SectorIndex}$ the input domain of the $\SectorIndex$-th azimuthal partition, defined as:
\begin{align}
	\InputDomain_{\Observer,\SectorIndex} &=
	\left\{ \State_{\Observer} \in \InputDomain_{\Observer} \,\mid\, 
	\Azimuth \in \left[\underbar{\Azimuth}_{\SectorIndex}, \overbar{\Azimuth}_{\SectorIndex}\right) \right\}, \quad m \in\{1, \ldots, \SectorCt\}. \label{eq:region_inputspace}
\end{align}

Leveraging~\assumpref{assump:noise_monotonicity} and~\assumpref{assump:angular_monotonicity_noise}, we conclude that the noise function $\TrueNoise$ is monotonic in each azimuthal partition $\InputDomain_{\Observer,\SectorIndex}$. 
Recalling Problem~\ref{prob:noisePrediction}, our goal is to learn a NN model whose
absolute approximation error admits a uniform bound on the operating set. The monotonicity of $\eta$ in each azimuthal partition $\InputDomain_{\Observer,\SectorIndex}$ suggest that one should adopt \emph{monotonic neural network} architectures~\cite{nolte2023expressive}:

\begin{definition}[Monotonic Neural Network]
\label{def:monotonic_neural_network}
A $\NN$ is {monotonically increasing} (resp. decreasing) {with respect to input $x_i$} if: 
\[
\frac{\partial \NN}{\partial x_i} \ge 0 \quad (\text{resp. } \frac{\partial \NN}{\partial x_i} \le 0).
\]
The network is {monotonic} if it is {monotonically increasing} (resp. decreasing) in each input $x_i$.

\end{definition}

Therefore, we reach our final design of the NN as a set of \emph{monotonic} NN submodels (one for each azimuthal partition),
yielding the monotonic NN-based noise model:
\begin{align}
\NN^{\textup{noise}}(\State_{\Observer}) &= \sum_{\SectorIndex=1}^{\SectorCt} \NN_{\!\SectorIndex}^{\textup{noise}}(\State_{\Observer} \mid \tilde{\Azimuth}_{\SectorIndex}) \cdot \mathds{1}_{\{\Azimuth \in \Azimuth_{\SectorIndex}\}}, \label{eq:neural_networks_by_regions}
\end{align}
where $\NN_{\!\SectorIndex}^{\textup{noise}}$ is the $\SectorIndex$-th sector-specific
monotonic NN-based modeling the noise within the azimuthal partition.

\subsection{Formal Guarantees on Monotonic NN-based Noise Model} 
\label{subsec:Formal_Guarantee}
In this section, we establish formal guarantees on the accuracy of the monotonic NN-based noise model $\NN^{\textup{noise}}$. Our objective is to generalize the \emph{empirical error} calculated after NN training into \emph{worst case error} bound that is guaranteed to hold across the entire input domain, even for data points that were not part of the NN training/testing.

To derive such worst case bound, we resort to the monotonicity assumptions (Assumption~\ref{assump:noise_monotonicity} and Assumption~\ref{assump:angular_monotonicity_noise}) as follows. First, consider any two data points  $\State_{\Observer}' = (v',\rho',h',r',\varphi')$ and $\State_{\Observer}'' = (v'',\rho'',h'',r'',\varphi'')$ randomly sampled from the same $\InputDomain_{\Observer,\SectorIndex}$. Assume that these two data points $\State_{\Observer}', \State_{\Observer}''$ respect the monotonicity of $\eta$, i.e.:
\begin{align}
    v' \ge v'', \; \rho' \ge \rho'', \; h' \le h'', \; r', \le r''.
\end{align}
Notice that these two points $\State_{\Observer}'$ and $\State_{\Observer}''$ can be seen as two antipodal vertices of a hypercube $\Hypercube_{\SectorIndex}$ defined as: $$\Hypercube_{\SectorIndex} = [v'', v']\times[\rho'',\rho']\times[h',h'']\times[r',r'']\times[\varphi', \varphi''].$$
Thanks to Assumptions~\ref{assump:noise_monotonicity} and~\ref{assump:angular_monotonicity_noise}, the following holds:
$$\forall \State_\Observer \in \Hypercube_{\SectorIndex} : \qquad \eta(\State_{\Observer}'') - \BoundedVarAzi \le \eta(\State_{\Observer}) \le \eta(\State_{\Observer}') + \BoundedVarAzi.$$
In other words, data points collected in a way that respects the noise monotonicity leads to a sampling lattice where the noise function within this hypercube is bounded by the values at the two antipodal vertices. Hence, without loss of generality, we assume the test dataset $\Dataset_{\Test}$ was collected from such a sampling lattice, i.e.: 

\begin{align}
	\Dataset_{\Test}
	= \bigcup_{\SectorIndex=1}^{\SectorCt} \Dataset_{\Test,\SectorIndex}
	= \bigcup_{i=1,j=1,k=1,\ell=1,\SectorIndex=1}^{\SpeedCt-1, \RpmCt-1, \VdistanceCt-1, \HdistanceCt-1, \SectorCt}
	\left\{\overbar{\State}_{\Observer,\SectorIndex}^{i j k \ell}, \uunderbar{\State}_{\Observer,\SectorIndex}^{i j k \ell} \right\},\label{eq:test_dataset}
\end{align}
where:
\begin{align}
\overbar{\State}_{\Observer,\SectorIndex}^{i j k \ell}
	&= (\Speed_{i+1}, \Rpm_{j+1}, \Vdistance_k, \Hdistance_l, \tilde{\Azimuth}_{\SectorIndex}), \nonumber \\
\uunderbar{\State}_{\Observer,\SectorIndex}^{i j k \ell} 
	&= (\Speed_i, \Rpm_j, \Vdistance_{k+1}, \Hdistance_{l+1}, \tilde{\Azimuth}_{\SectorIndex}),\notag
	\label{eq:hypercube_max_min}
\end{align} 
with $v_{i+1} \ge v_i, \rho_{j+1} \ge \rho_j, h_{k+1} \le h_k, r_{l+1} \le r_l$ and $(i,j,k,l) \in \IndexSet$ where 
$\IndexSet = [1 \!:\! \SpeedCt] \!\times\! [1 \!:\! \RpmCt] \!\times\! [1 \!:\! \VdistanceCt] \!\times\! [1 \!:\! \HdistanceCt] $ is the index set, and
$\SpeedCt$, $\RpmCt$, $\VdistanceCt$, and $\HdistanceCt$ are the number of grid points in the dimensions $\Speed$, $\Rpm$, $\Vdistance$, and $\Hdistance$, respectively. Such test dataset gives ruse to the following set of hypercubes:
\begin{align*}
	\Hypercube_{\SectorIndex}^{i j k \ell} = 
		\left[\Speed_i, \Speed_{i+1}\right] \!\times\!
		\left[\Rpm_j, \Rpm_{j+1}\right] \!\times\!
		\left[\Vdistance_k, \Vdistance_{k+1}\right] \!\times\!
		\left[\Hdistance_\ell, \Hdistance_{\ell+1}\right] \!\times\!
		\{ \tilde{\Azimuth}_{\SectorIndex} \}.
\end{align*}

Given the structure of the trained monotonic NN-based noise model $\NN^{\textup{noise}}$ in \ref{eq:neural_networks_by_regions} and the definition of $\Dataset_{\Test}$ above, we state our main theoretical result as follows.

\begin{theorem} \label{thm:error_bound}
Given a trained
monotonic NN-based noise model $\NN^{\textup{noise}}$ and a bounded input domain $\InputDomain$, let $\Dataset_{\Test}$ be as defined in~\eqref{eq:test_dataset}.
Then under the Noise Monotonicity \assumpref{assump:noise_monotonicity}, the following bound on the prediction error holds:
\begin{align*}
	\max_{\State_{\Observer} \in \InputDomain_{\Observer}} |
		\TrueNoise\left(\State_{\Observer}\right)
		- \NN^{\textup{noise}}\left(\State_{\Observer}\right)
	|	\leq \ErrorBound
\end{align*}
where the error bound $\ErrorBound$ is defined as:%
\begin{align*}
	\ErrorBound &=
		\max_{\SectorIndex \in [1:\SectorCt]}\;\ErrorBound_{m}\\
    &=
		\max_{\SectorIndex \in [1:\SectorCt]}\;
	\max_{\State_{\Observer} \in \Dataset_{\Test,\SectorIndex}} \;
	\left\{ \begin{aligned}
	1 & +
	\underbrace{\big| \TrueNoise(\State_{\Observer}) - \ConstC_{\SectorIndex} \big|}_{\Term_1}
	\\
	& 
	+ \underbrace{ \big| \NN_{\!\SectorIndex}^{\textup{noise}}(\State_{\Observer}) - \ConstD_{\SectorIndex} \big| }_{\Term_2}\\
	& + \underbrace{\big| \ConstC_{\SectorIndex} - \ConstD_{\SectorIndex} \big|}_{\Term_3}
	\end{aligned} \right\},
\end{align*}
\begin{align*}
	\text{and } \quad
	\ConstC_{\SectorIndex} &= \frac{\TrueNoise(\overbar{\State}_{\Observer,\SectorIndex}^{i j k \ell}) + \TrueNoise(\uunderbar{\State}_{\Observer,\SectorIndex}^{i j k \ell})}{2},\\
	\ConstD_{\SectorIndex} &= \frac{\NN^{\textup{noise}}_{\SectorIndex}(\overbar{\State}_{\Observer,\SectorIndex}^{i j k \ell}) +\NN^{\textup{noise}}_{\SectorIndex}(\uunderbar{\State}_{\Observer,\SectorIndex}^{i j k \ell})}{2}.
\end{align*}
\end{theorem}
\noindent\emph{Proof.} See Appendix~\ref{app:proofs}.

\thmref{thm:error_bound} establishes that, the worst case error between any trained monotonic NN-based noise model $\NN^{\textup{noise}}$
and the unknown noise function $\TrueNoise$,
is upper bounded by $\ErrorBound$. Note that $\ErrorBound$ holds for all states in $\InputDomain_{\Observer}$ even those who are not present in the training or test dataset. Moreover, $\ErrorBound$ can be computed by evaluating the neural network and the noise function at a discrete set of points namely the vertices of each hypercube $\Hypercube^{i j k \ell}_{\SectorIndex}$ (i.e., the points in the dataset $\Dataset_{\Test}$).
It also indicates that the bound $\ErrorBound$ comprises three key terms:
\begin{itemize}
	\item $\Term_1$, which captures the variation of the noise levels within the hypercube $\Hypercube_{\SectorIndex}^{i j k \ell} $;
	\item $\Term_2$, which quantifies the fluctuation of the neural network within each hypercube, serving as an indicator of the model smoothness; and
	\item $\Term_3$, which captures the systematic bias between the true noise function $\TrueNoise$ and the neural network prediction $\NN^\textup{noise}_{\SectorIndex}$ within each hypercube, serving as a measure of training accuracy relative to the ground truth.
\end{itemize}
These terms quantify the different sources of the neural network prediction error and help determine when finer sampling strategies are needed to ensure high fidelity.
Note that $\ConstC_{\SectorIndex}$ and $\ConstD_{\SectorIndex}$ are hypercube-specific constants that represent the average true and predicted noise levels, respectively, at the corners of the hypercube.

\section{Active Sampling for Optimal Dataset Acquisition}
\label{sec:active_sampling}

While~\thmref{thm:error_bound} holds for any test dataset $\Dataset_{\Test}$, we focus in this section on how to select $\Dataset_{\Test}$ that yields to an error bound $\delta$ that is below a user-specified threshold. Indeed, one can reduce the error bound $\delta$ by adding more samples to $\Dataset_{\Test}$, however, and as motivated before, collecting one data point is computationally demanding. 

Hence, motivated by~\thmref{thm:error_bound} and by the high cost of obtaining noise data in terms of time and resources, we design an active-sampling algorithm that, for any given input domain, iteratively collects an optimal number of samples in $\Dataset_{\Test}$ while guaranteeing a bound on the absolute error bound $\ErrorBound$ defined in ~\thmref{thm:error_bound}. 
The algorithm recursively subdivides hypercubes until the local metric $\Term_1$ falls below a user-specified target. This process both reduces $\Term_1$ and collects only those $\Dataset_{\Test}$ samples that most effectively reduce $\Term_1$. Consequently, it also helps decrease $\Term_2$ during training. This refinement not only lowers the local variation in the noise but also tightens the final error bound $\ErrorBound$.
Specifically, for any point inside a hypercube in the generated dataset,
the variation of the noise function $\Term_1$ is provably less than a user-specified threshold.
This guarantee ensures that models trained or validated on this dataset
are not subject to hidden variations in noise beyond the specified bound.

\subsection{Active Sampling Algorithm Overview}

\algref{alg:active_sampling} outlines the procedure for generating the dataset
for the azimuthal sector $\SectorIndex$. For simplicity—both in algorithmic exposition and in the visualization of \figref{fig:hypercube_subdivision}—we present the procedure in three dimensions; without loss of generality, the same construction extends directly to higher-dimensional inputs (e.g., the full 4-D $(\Speed,\Rpm,\Vdistance,\Hdistance)$ case). The algorithm operates by recursively partitioning the input space
defined over $\Speed$, $\Rpm$, and $\Vdistance$,
while keeping $\Hdistance$ and $\Azimuth$ fixed. Starting with the input bounds $(\underbar{\Speed}, \overbar{\Speed})$, $(\underbar{\Rpm}, \overbar{\Rpm})$, and $(\underbar{\Vdistance}, \overbar{\Vdistance})$, 
the root hypercube $\Hypercube_0$ is initialized using its two diagonal corners.
Next, the maximum variation in noise values within the hypercube is computed by evaluating the noise function $\TrueNoise$ at the corners of the hypercube.
If the variation in noise levels is below the specified threshold $\BoundedVarAct$,
the hypercube is accepted, and the two corners
and their corresponding noise levels are added to the dataset $\Dataset_{\Test}$.
If the variation exceeds the threshold,
the hypercube is then subdivided into eight smaller hypercubes as shown in~\figref{fig:hypercube_subdivision}.
The subdivision is performed by calculating the midpoint of each dimension,
then constructing the new hypercubes using the midpoints and the original corners.

\begin{algorithm}[t]
\caption{Active Sampling with Bounded Noise Variation}
\label{alg:active_sampling}
\KwIn{%
  Input bounds
  $(\underbar{\Speed}, \overbar{\Speed}), (\underbar{\Rpm}, \overbar{\Rpm}), (\underbar{\Vdistance}, \overbar{\Vdistance})$ for $\InputDomain_\Observer$;\\ \phantom{\textbf{Input:~}}%
  Input values $\Hdistance$, $\Azimuth \in \Azimuth_\SectorIndex$;\\ \phantom{\textbf{Input:~}}%
  Noise function $\TrueNoise$.
}
\KwOut{Certified test dataset $\Dataset_{\Test}$}
$\Hypercube^{0} \gets [\underbar{\Speed}, \overbar{\Speed}] \times [\underbar{\Rpm}, \overbar{\Rpm}] \times [\underbar{\Vdistance}, \overbar{\Vdistance}] \times \{ \Hdistance \} \times \{ \Azimuth \}$\;
$\Queue \gets [\Hypercube^{0}]$;\,\, $\Dataset_{\Test} \gets \emptyset$;\,\, $i \gets 0$\;
\While{$\Queue \neq \EmptyWord$}{
	$\Hypercube^{i+1} \gets \Front(\Queue);\,\, \Dequeue(\Queue)$\;
	$\State_1 \gets \overbar{\State}_{\Observer,\SectorIndex}^{i+1};\,\,$
	$\underbar{\NoiseValue} \gets \TrueNoise(\State_1)$\;
	$\State_3 \gets \uunderbar{\State}_{\Observer,\SectorIndex}^{i+1};\,\,$
	$\overbar{\NoiseValue} \gets \TrueNoise(\State_3)$\;
	\eIf{$\overbar{\NoiseValue} - \underbar{\NoiseValue} \leq \BoundedVarAct$}
	{
		$\Dataset_{\Test} \gets \Dataset_{\Test}
		\cup \left\{\left(\State_1, \underbar{\NoiseValue}\right), \left(\State_3, \ooverbar{\NoiseValue}\right)\right\}$\;
	}%
	{
		Find subdivision corners:
			$\left(\Speed_2,\Rpm_2, \Vdistance_2\right) \gets
			\frac{1}{2}\cdot\left( \Speed_1 + \Speed_3,\; \Rpm_1 + \Rpm_3,\; \Vdistance_1 + \Vdistance_3\right)$\;
		Subdivide $\Hypercube$ into 8 child hypercubes:\\
		\For{$ w \gets 0$ \KwTo $3$}{
			\ForEach{$(j,k,\ell) \in \{1,2\}^3$, $j\!+\!k\!+\!\ell=w$}{
				$\Enqueue(\Queue, \allowbreak 
				\left\{
					[\Speed_{j}, \Speed_{j+1}],
					[\Rpm_{k}, \Rpm_{k+1}],
					[\Vdistance_{\ell}, \Vdistance_{\ell+1}]
				\right\})$\;
			}
		}
	}
  $i \gets i + 1$\;
}
\Return{$\Dataset_{\Test}$}\;
\end{algorithm}

\begin{figure}[t]
    \centering
    \includegraphics[width=\columnwidth]{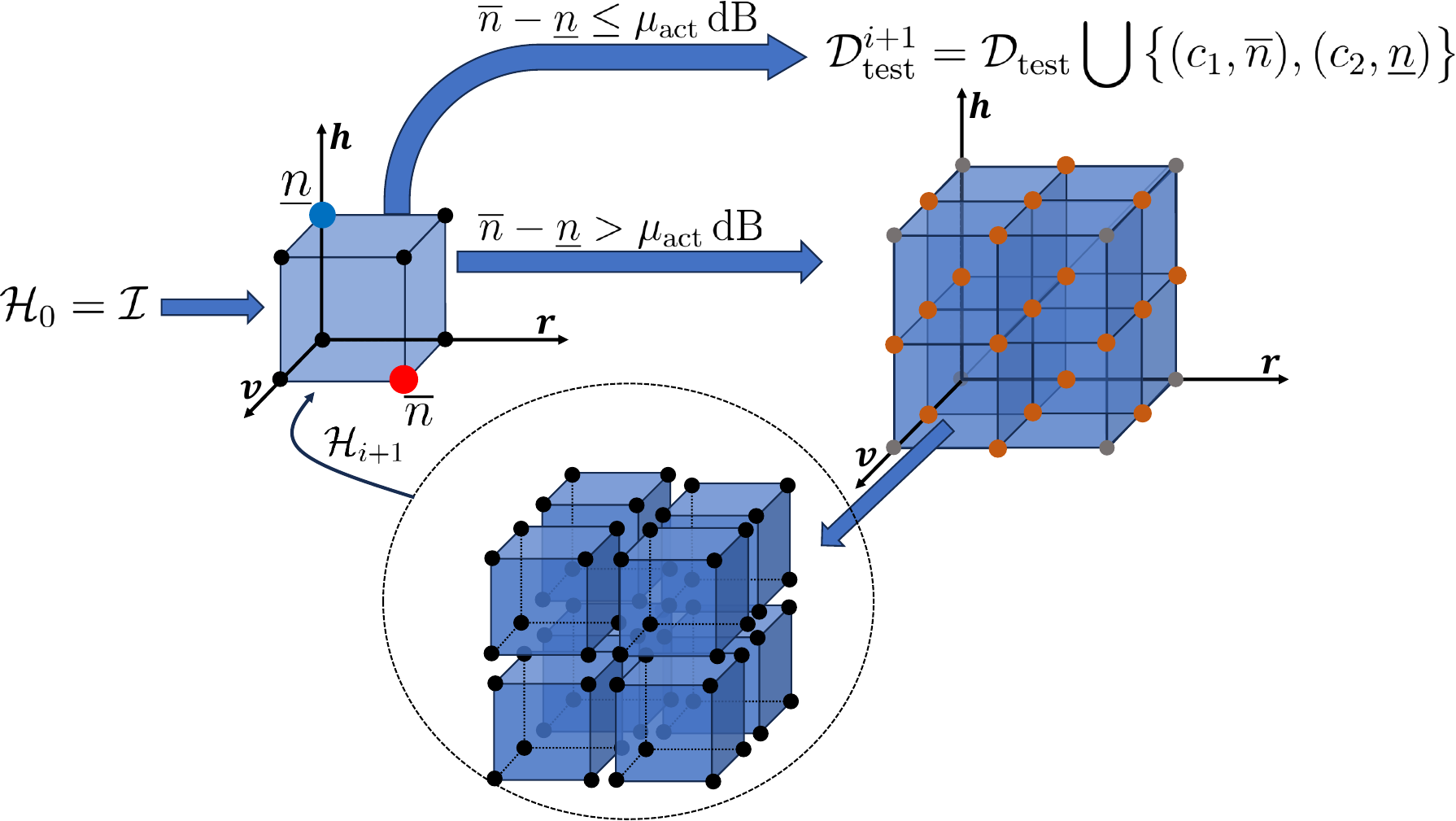}
    \caption{Flow Diagram of Active Sampling}
    \label{fig:hypercube_subdivision}
\end{figure}

To keep track of the hypercubes to be processed,
the algorithm uses a FIFO queue $\Queue$ that operates in a breadth-first manner.
At each iteration, the algorithm dequeues the hypercube at the front,
evaluates its noise variation, and either certifies it or subdivides it into smaller hypercubes.
These child hypercubes are then added to the queue following a ranked lattice traversal with lexicographic refinement,
where the rank is defined by the sum of the subdivision indices $(j,k,\ell)$.
This queue-based design guarantees that the algorithm systematically refines only those regions where the noise function exhibits high variability, while efficiently avoiding unnecessary evaluations in regions already deemed smooth.

Since the input domain is bounded and each subdivision reduces the hypercube size, it follows that the algorithm is guaranteed to terminate as every hypercube eventually becomes small enough to satisfy the noise threshold.
The resulting dataset carries a certified approximation of the noise function over the input domain, with explicit guarantees on maximum variation within each azimuthal partitioning sector.
This certified property is critical for providing bounded guarantees for a trained, monotonic NN-based noise model.

\section{Noise-Aware Motion Planning}
\label{sec:motion_planning}
In this section, we design a noise-aware motion-planning algorithm to address~\probref{prob:motion_planning}.
The algorithm is a variant of kinodynamic RRT*~\cite{webb2013kinodynamic}
that incorporates the noise constraints imposed by the NAZs.
To ensure compliance with noise ordinances, the algorithm relies on the ability to accurately and timely predict the noise levels at all observers given the eVTOL's state.
This is enabled by the certified noise model described in~\secref{sec:noise}. 
\subsection{Noise-Aware Kinodynamic RRT* Algorithm}
As a pre-processing step, our algorithm starts with \emph{tightening} the noise specifications in each NAZ to account for the worst case prediction error $\delta$. That is, we define the the \emph{$\ErrorBound$-tightened} NAZ set $\NAZmar$ as:
$$\NAZ^{\delta}_j = \left(\NAZ_{j}.\Observer, \NAZ_{j}.\InsNoiseThreshold-\ErrorBound, \NAZ_{j}.\AvgNoiseThreshold-\ErrorBound, \NAZ_{j}.{\Window}, \NAZ_{j}.\Ordinance\right).$$ 
Note that both the error bound $\delta$ and the noise thresholds $(\InsNoiseThreshold,\AvgNoiseThreshold)$ are all in dBA units (logarithmic units) and hence the subtraction $\InsNoiseThreshold-\ErrorBound$ and $\AvgNoiseThreshold-\ErrorBound$ should be computed as follows:
$$\overline{\TrueNoise} - \delta = 10\log_{10}\left(10^{\overline{\TrueNoise}/10}-10^{\ErrorBound/10}\right),$$
where $\overline{\TrueNoise}$ in the equation above can be either $\InsNoiseThreshold$ or $\AvgNoiseThreshold$.

\algref{alg:kino_rrtstar} presents the procedure for generating a noise-aware motion plan.
The algorithm incrementally builds a tree of trajectories that account for the noise constraints imposed by the \emph{$\ErrorBound$-tightened} NAZ set $\NAZmar$.
Then, it returns a noise-compliant motion plan $\Plan$ as a sequence of eVTOL states.
Each node $q$ in the tree encodes a candidate eVTOL state $\State$, its parent node,
the cost to reach this state, and the time at which the state was reached.
The tree $\Tree$ is initialized with the node $\Node_\Start$ corresponding to the start state $\State_{\Start}$ and time $t_{\Start}$.
Next, each iteration of the algorithm starts with sampling a random state $\State_\Rand \sim \Uniform(\InputDomain_\Observer)$.
The node $\Node_\Near \in \Tree$ with the nearest state to $\State_\Rand$ is then determined by using $\fnGetKinoDist$ (\linesref{line:fnGetKinoDist_start}{line:fnGetKinoDist_end} in \algref{alg:kino_rrtstar}) to evaluate the kinodynamic distance between $\State_\Rand$ and the state $\Node.\State$ for all $\Node \in \Tree$ (\lineref{line:nearest_node} of~\algref{alg:kino_rrtstar}).
Using $\Node_\Near.\State$ and $\State_\Rand$ as start and candidate end states,
the function $\fnSteer$ is called to attempt to find a new state $\Node_\New$ that is reachable from $\Node_\Near$.
Specifically, the function $\fnSteer$ computes multiple candidate trajectories from $\Node_\Near.\State$ to $\State_\Rand$, the feasibility of each is checked against both the kinodynamic and noise-based constraints.
The new node $\Node_\New$ is generated based on the feasible trajectory (if any) with the lowest cost by using $\fnGetCost$ (\algref{alg:getCost}), where $\fnGetCost$ returns the accumulated cost 
$J' = J + c$ and $c$ is a $\CostWeight$-weighted sum of
(i) segment length and kinodynamic deviation,
(ii) control-effort/smoothness penalties on $\Delta v$ and $v_{\Vdistance}$,
(iii) speed $v$ reward, and
(iv) proximity-shaping terms that activate near the goal.
Once added to the tree, the algorithm attempts to rewire the tree by connecting $\Node_\New$ to any nearby nodes at a lower cost.
Finally, the best node $\Node_\Best$ is updated whenever a new node is within distance $\epsilon_{\Goal}$ from the goal state $\State_{\Goal}$ and has a lower cost than the current best node.
The algorithm uses $\fnExtractPath$ to extract and return the motion plan $\Plan$.

\begin{algorithm}[!t]
\caption{Noise-Aware Kinodynamic RRT*}
\label{alg:kino_rrtstar} 
\KwIn{
	$\State_{\Start},\State_{\Goal}$, start and goal states; \\
	\hspace{3em}$\InputDomain$, bounded input domain; \\
	\hspace{3em}$\NoiseModel \gets \NoiseModel^{\textup{noise}}$, noise model; \\
	\hspace{3em}$\NAZmar$, $\ErrorBound$-tightened NAZ set; \\
	\hspace{3em}$\CostWeight$, cost weights; \\ 
	\hspace{3em}$N_{\max}$, step, goal bias, search radius, max iter; \\
	\hspace{3em}$\epsilon_{\Goal}$, goal tolerance;
}
\KwOut{%
  $\Plan \in \Planset \cup \{\emptyset\}$, motion plan or no feasible path
}

\SetKwProg{KwFnRRTs}{function}{\,:}{end}
\KwFnRRTs{%
	$ \mathtt{getMotionPlan}\,
	(\State_{\Start}$, $\State_{\Goal}$, $\InputDomain_\Observer$, $\NoiseModel$,
	$\NAZmar$, $\CostWeight$,
  $\Delta t$, $N_{\max}$, $\epsilon_{\Goal})$
}{
	$\Node_\Start \gets \left(\State_{\Start}, \emptyset, 0, t_{\Start} \right)$%
		\tcp*[r]{(state, parent, cost, time)}
	$\Tree \gets \{ \Node_\Start \}$;\,\,$\Node_\Best \gets (\emptyset, \emptyset, \infty, \infty)$\;
	\For{$n = 1$ to $\IterationCt$}{
		$\State_\Rand \gets \fnSampleState(\InputDomain)$\;

		$\Node_\Near \gets \argmin_{\Node \in \Tree} \fnGetKinoDist(\Node.\State, \State_\Rand)$\;%
		\label{line:nearest_node}

		$\Node_\New \gets \fnSteer(\Node_\Near, \State_\Rand, \InputDomain, \NoiseModel, \NAZmar, \CostWeight)$\;

		\lIf{$\Node_\New.\State = \emptyset$}{
			\textbf{continue}}%
				
		$\Tree \gets \Tree \cup \{\Node_\New\}$\;

		\ForEach{$\Node \in \Tree\setminus \{\Node_\New\}$}{
			$\Node' \gets \fnSteer(\Node, \Node_\New.\State, \InputDomain, \NoiseModel, \NAZmar, \CostWeight)$\;
			\lIf{$\Node'.\State \neq \emptyset \land
						\Node'.\Cost < \Node.\Cost $}{
						$\Node \gets \Node'$
					}
		}
		\lIf{$\fnDist(\Node_\New.\State, \State_\Goal) < \epsilon_{\Goal} \land \Node_\New.\Cost < \Node_\Best.\Cost$}{
				$\Node_\Best \gets \Node_\New$}
	}

	\lIf{$\Node_\Best.\State = \emptyset$}{
			\Return{$\emptyset$}
	}
	\lElse{\Return{$\fnExtractPath(\Node_\Best)$}}
}
\vspace{0.6em}
\SetKwProg{KwFnGetKinoDist}{function}{\,:}{end}
\KwFnGetKinoDist{$\fnGetKinoDist\,(\State,\State')$%
}{\label{line:fnGetKinoDist_start}
{$(\PosX,\PosY,\Vdistance,\Heading) \gets (\State.\PosX,\State.\PosY,\State.\PosZ,\State.\Heading) $\;
$(\PosX',\PosY',\Vdistance',\Heading') \gets (\State'.\PosX,\State'.\PosY,\State'.\PosZ,\State'.\Heading) $}\;
	\Return{$\!\sqrt{(\PosX \!\!-\! \PosX')^2 \!+\! (\PosY \!\!-\! \PosY')^2 \!+\! (\Vdistance \!\!-\! \Vdistance')^2 \!+\!
	\tfrac{1}{2\pi}(\Heading \!\!-\! \Heading')^2}$}\;%
	\label{line:fnGetKinoDist_end}
}

\vspace{0.6em}
\SetKwProg{KwFnSteer}{function}{\,:}{end}
\KwFnSteer{$\fnSteer\,(\Node_\Near, \State_\Rand, \InputDomain, \NoiseModel, \NAZmar,\CostWeight)$%
}{
	$\Node_\New \gets (\emptyset, \emptyset, \infty, t_\Near)$\;
	\For{$i = 1$ to $\AttemptCt$}{
		$(\Speed, \Vdistance, \Delta \Heading) \gets
		\Uniform([\underbar{\Speed}, \overbar{\Speed}] \times [\underbar{\Vdistance}, \overbar{\Vdistance}] \times [\Delta\underbar{\Heading}, \Delta\overbar{\Heading}])$\; \label{line:fnSteer_sample}
		$(\State',t) \gets \fnSimulate(\Node_\Near, \State_\Rand, \Speed, \Vdistance, \Delta \Heading, \Delta t)$\;
		\ForEach{$\NAZ \in \NAZmar$}{%
			$\NoiseValue_{\NAZ.\Observer} \gets \fnPredictNoise(\NoiseModel, \State', \NAZ.\Observer)$\;
			$\NoiseValue_{\Average,\NAZ.\Observer} \gets \fnSlidingAverage(\NoiseValue_{\NAZ.\Observer}, \NAZ.\Window)$;
		}
		\If{$
					\land_{\NAZ \in \NAZmar} (\NoiseValue_{\NAZ.\Observer} \!\leq\! \NAZ.\InsNoiseThreshold \allowbreak\land \NoiseValue_{\Average,\NAZ.\Observer} \!\leq\! \NAZ.\AvgNoiseThreshold)
					\allowbreak\land
					\fnGetCost(\Node_\Near.\State, \State',\State_\Goal,\CostWeight) < \Node_\New.\Cost
					\allowbreak\land
					\lnot\fnDetectCollision(\Node_\Near.\State, \State')
			$\label{line:noise_comparison}}{
				$\Node_\New \gets (\State', \Node_\Near, \fnGetCost(\Node_\Near.\State, \State',\State_\Goal,\CostWeight), t)$;
		}
	}
	\Return{$\Node_\New$}
} 

\end{algorithm}
A major difference from the standard kinodynamic RRT* algorithm is its incorporation of noise constraints into the planning process.
A candidate state is considered feasible if the trajectory (a) is feasible with respect to the eVTOL's kinematics, (b) does not collide with other eVTOLs, and (c) does not violate the noise constraints at any observer.
At each NAZ, the function $\fnPredictNoise$ uses the observer coordinates $\NAZ.\Observer$ and the eVTOL candidate state $\State'$ to compute $\State'_\Observer$,
\ie the eVTOL's state with respect to the observer, and further use it as input to the monotonic NN-based noise model $\NoiseModel^{\textup{noise}}$ to predict $\NoiseValue_{\NAZ.\Observer}$,
\ie the current noise level at that observer.
Next, the function $\fnSlidingAverage$ keeps track of the past noise levels at all observers, uses $\NoiseValue_{\NAZ.\Observer}$ to update this history, and computes the equivalent noise $\NoiseValue_{\Average,\NAZ.\Observer}$.
The noise values $\NoiseValue_{\NAZ.\Observer}$ and $\NoiseValue_{\Average,\NAZ.\Observer}$ are then compared against the noise thresholds $\NAZ.\InsNoiseThreshold$ and $\NAZ.\AvgNoiseThreshold$, respectively,
for each NAZ (\lineref{line:noise_comparison} of~\algref{alg:kino_rrtstar}).
Both conditions become part of the feasibility check for the candidate state to make sure that the noise ordinance defined in~\eqref{eq:ordinance_conditions} is satisfied.

\subsection{Physics-Based Sampling for Motion Planning}
\label{sec:motion_planning:pbs}

We now exploit the monotonicity property of the noise model to improve the sampling efficiency of the steering function $\fnSteer$.
Recall that control inputs $(\Speed, \Vdistance, \Delta \Heading)$ are sampled uniformly from the admissible set $\InputDomain_\Observer$ (\lineref{line:fnSteer_sample} of~\algref{alg:kino_rrtstar}).
In a \emph{uniform random sampling} (URS) strategy, the control inputs are drawn uniformly from the entire admissible set until a sample both (i) minimises distance to the target state and (ii) satisfies the noise constraints in~\defref{def:naz}.
This admissible set remains unchanged throughout the planning process, which can lead to inefficient sampling and long planning times, especially in scenarios with strict noise constraints.

To address this, we propose a \emph{physics-based sampling} (PBS) strategy that leverages the monotonicity property of the noise model.
When a tuple $(\Speed, \Vdistance, \Delta \Heading)$ fails either noise test, all tuples with higher $\Speed$ and lower $\Vdistance$ are pruned from the search space.
Subsequent samples are restricted to the remaining domain, which shrinks over time as more samples are drawn.
This approach significantly reduces the number of samples required to find a feasible trajectory, especially in scenarios with strict noise constraints as shown later in~\secref{sec:experiments}.
\algref{alg:steer_pbs} presents the steering function $\fnSteerPBS$ that implements this procedure.
Note that the pruning of the search space (\lineref{line:prune_search_space} of~\algref{alg:steer_pbs}) is only applied when the noise constraints are violated.
While the noise model is monotonic with respect to the rotor RPM $\RPM$, the pruning is only
applied to the control inputs $\Speed$ and $\Vdistance$, since the change in the noise level is insignificant as we show later in~\secref{sec:experiments}, specifically at~\figref{fig:subfig_2}.

\begin{algorithm}[t]
\caption{Steering Function for Noise-Aware Kinodynamic RRT* using Physics-Based Sampling}
\label{alg:steer_pbs} 
\vspace{0.6em}
\SetKwProg{KwFnSteer}{function}{\,:}{end}
\KwFnSteer{$\fnSteerPBS\,(\Node_\Near, \State_\Rand, \InputDomain, \NoiseModel, \NAZmar,\CostWeight)$%
}{
	$\Node_\New \gets (\emptyset, \emptyset, \infty, t_\Near)$\;
	\For{$i = 1$ to $\AttemptCt$}{
		$(\Speed, \Vdistance, \Delta \Heading) \gets
		\Uniform([\underbar{\Speed}, \overbar{\Speed}] \times [\underbar{\Vdistance}, \overbar{\Vdistance}] \times [\Delta\underbar{\Heading}, \Delta\overbar{\Heading}])$\;
		$(\State',t) \gets \fnSimulate(\Node_\Near, \State_\Rand, \Speed, \Vdistance, \Delta \Heading, \Delta t)$\;
		\ForEach{$\NAZ \in \NAZmar$}{%
			$\NoiseValue_{\NAZ.\Observer} \gets \fnPredictNoise(\NoiseModel, \State', \NAZ.\Observer)$\;
			$\NoiseValue_{\Average,\NAZ.\Observer} \gets \fnSlidingAverage(\NoiseValue_{\NAZ.\Observer}, \NAZ.\Window)$;
		}
		\If{$
					\land_{\NAZ \in \NAZmar} (\NoiseValue_{\NAZ.\Observer} \!\leq\! \NAZ.\InsNoiseThreshold \allowbreak\land \NoiseValue_{\Average,\NAZ.\Observer} \!\leq\! \NAZ.\AvgNoiseThreshold)
			$\label{line:noise_comparisonb}}{
				\lIf{$
					\allowbreak\land
					\fnGetCost(\Node_\Near.\State, \State',\State_\Goal,\CostWeight) < \Node_\New.\Cost
					\allowbreak\land
					\lnot\fnDetectCollision(\Node_\Near.\State, \State')
				$}{
				$\Node_\New \gets (\State', \Node_\Near, \fnGetCost(\Node_\Near.\State, \State',\State_\Goal,\CostWeight), t)$;
				}
			}
		\lElse{
			$\overbar{\Speed} \gets \Speed$; $\underbar{\Vdistance} \gets \Vdistance$ \label{line:prune_search_space}%
		}
	}
	\Return{$\Node_\New$}
} 
\end{algorithm}

\begin{algorithm}[t]
\caption{Cost Evaluation Function for Node}
\label{alg:getCost}
\SetKwProg{KwFnGetCost}{function}{\,:}{end}
\KwFnGetCost{$\fnGetCost\!\left(\State,\ \State',\State_\Goal,\CostWeight\right)$%
}{
    $(\omega_{\text{prox}}, \omega_{\text{speed}}, \omega_d,\omega_q,Q_\Vdistance,R_{v},R_{v_\Vdistance},q_\Vdistance,q_v) \gets \CostWeight$ \label{line:CW}
    
    $x \gets \State.x$;\quad
    $y \gets \State.y$;\quad
    $\Vdistance \gets \State.z$;\\
    $x' \gets \State'.x$;\quad
    $y' \gets \State'.y$;\quad
    $\Vdistance' \gets \State'.z$;\\
    $x_g \gets \State_{\Goal}.x$;\quad
    $y_g \gets \State_{\Goal}.y$;\quad
    $\Vdistance_g \gets \State_{\Goal}.z$;\\
    
    $J\gets q.\Cost$;

    $d_{\text{kin}} \gets \fnGetKinoDist(\State,\State')$;

    $\Delta v \gets \State'.v - \State.v$;\quad $v_\Vdistance \gets (\Vdistance'-\Vdistance)/\Delta t$;

    $d_\Goal \gets \sqrt{(\PosX'-\PosX_\Goal)^2+(\PosY'-\PosY_\Goal)^2+(\Vdistance'-\Vdistance_\Goal)^2}$;\\
    $\omega_{\text{prox}} \gets \max\!\big(0,\ D_0 - d_{\Goal}\big)$;
    
    $c \gets
      \omega_d\,d_{\text{kin}}
      + \omega_q\!\left(Q_\Vdistance \Vdistance' + R_{v}|\Delta v| + R_{v_\Vdistance}|v_\Vdistance|\right)
      - \omega_{\text{speed}}\,v'
      + \omega_{\text{prox}}\!\left(q_\Vdistance \Vdistance' + q_vv'\right)$;

    \Return{$J + \Cost$}%
    \label{line:fnGetCost_end}
    }
\end{algorithm}

\subsection{Correctness of the Motion Planning Algorithm}
The following proposition captures the correctness guarantees of \algref{alg:kino_rrtstar}.
\begin{proposition}
    Motion plans $\Plan^\Horizon$ generated by \algref{alg:kino_rrtstar} satisfies the constraints in Problem~\ref{prob:motion_planning}.
\end{proposition}
The proof is a straight forward consequence of the certified bound $\delta$ along with the soundness of the kinodynamic RRT* algorithm.

\section{Experimental Results and Evaluation}
\label{sec:experiments}

In this section, we present the numerical results of three sets of experiments designed to validate our proposed framework for noise-aware motion planning.
The first set validates the correctness of the monotonic NN-based noise model;
specifically, we verify the correctness of the monotonicity assumption and the theoretical upper bound on the absolute error of the noise model.
The second set of experiments evaluates the active sampling strategy for noise prediction, where we analyze the partial derivatives of the noise with respect to each input variable.
The third set evaluates the performance of the noise-aware motion planner in both single- and multi-eVTOL scenarios,
and demonstrates the effectiveness of the monotonic-aware active sampling compared to a uniform sampling strategy.

\subsection{Experimental Setup}

To train the monotonic NN-based noise model, we generate the training and testing datasets
by running a physics-based, high-fidelity simulations of the eVTOL.
We use the FLOWUnsteady~\cite{flow_unsteady} simulation suite to model the unsteady aerodynamic behavior of an eVTOL under various operating conditions.
Each run simulates a steady-state flight conditions, where the eVTOL speed, rotor RPM, and altitude are varied
using the following levels:
		$\Speed \in \{20, 30, 40, 50, 60\}~\SI{}{\meter\per\second}$,
		$\RPM \in \{500, 600, 700\}~\SI{}{\rpm}$,
		$\Vdistance \in \{50, 100, \ldots, 450\}~\SI{}{\meter}$,
resulting in a total of $5 \times 3 \times 9 = 135$ treatment conditions.
For each treatment $(\Speed, \RPM, \Vdistance)$, the PSU-WOPWOP~\cite{psu_wopwop} simulation framework is used to simulate the Overall Sound Pressure Levels (OASPLs) in A-weighted decibels (dBA) ---
a measurement of the intensity of acoustic noise calculated as the decibel equivalent of the root sum square pressure. The OASPL represents the total acoustic energy of the entire spectrum, i.e., 
$$
\TrueNoise_{\Observer} = \text{OASPL} = 10 \log_{10} \left( \frac{\overbar{p^2}}{p_0^2} \right),
\overbar{p^2} = \frac{1}{T} \sum_{i=1}^{T} p_i^2,
$$
where $p_0 = 20 \times 10^{-6}$ (reference sound pressure) and $p_i$ is A-weighted pressure
--- the measured sound pressure level of a noise source, but with the frequencies adjusted to reflect how the human ear perceives loudness essentially giving more weight to mid-range frequencies and less to very low and very high frequencies.

The noise $\TrueNoise(\Speed,\Rpm,\Vdistance,\Hdistance,\Azimuth)$ was computed 
at a grid of $41 \times 41$ ground-level observer $\Observer$ locations, where observers were placed at $\SI{100}{\meter}$ intervals in both the $X$ and $Y$ directions,
spanning a total area of $\SI{4}{\kilo\meter} \times \SI{4}{\kilo\meter}$.
In total, $135\times41\times41=\SI{226935}{}$ noise samples were obtained.
Each sample is associated with the flight-condition and relative-coordinate tuple $(\Speed,\Rpm,\Vdistance,\Hdistance,\Azimuth)$—with the eVTOL cruising with a fixed leftward heading—to form the training and test datasets.

\subsection{Correctness of the Noise Prediction Model}
\label{sec:experiments:noise_correctness}

\experiment{Correctness of the Monotonicity Assumption}
This experiment aims to validate the correctness of the monotonicity assumption of the noise prediction model described in~\assumpref{assump:noise_monotonicity}.
We do this by numerically analyzing the partial derivative of the noise function with respect to each of the input variables $\Speed$, $\RPM$, and $\Vdistance$, using the data collected from the simulation.
\figref{fig:noise_contour} shows the noise contours for three different levels of each variable while keeping the other two fixed.
Generally, the noise contours show a monotonic increase in the noise function as the speed (\figref{fig:subfig_1}) and rotor RPM (\figref{fig:subfig_2}) increase, and a monotonic decrease as the altitude (\figref{fig:subfig_3}) increases.
As shown in~\figref{fig:monotonicity_test}, the results of analyzing the numerical partial derivatives confirm that the monotonicity assumption holds for $>\SI{99.98}{\percent}$ of the noise samples
($n = \SI{226935}{}$).
The remaining $<\SI{0.02}{\percent}$ of the samples can be attributed to numerical errors in the simulation that manifest as sudden drops in the noise levels.

\begin{figure*}[t]\setlength{\subfigcapskip}{-0.5em}
    \centering
		\setlength{\FigHeight}{2.25in}
		\setlength{\FigWidth}{2.6in}
		\setlength{\FigXGap}{-0.2in}
		\pgfplotsset{
			every axis/.append style={
				height=\FigHeight, width=\FigWidth, enlargelimits=false,
				title style={yshift=-0.7em, xshift=0em},
				xlabel={$\PosX$ (\si{\meter}) $\times 10^3$}, ylabel={$\PosY$ (\si{\meter}) $\times 10^3$},
				xlabel style={yshift=0.5em}, ylabel style={yshift=-0.2em},
				colormap/viridis,	colorbar=false,
				colorbar style={ytick={10,15,20,25,30,35,40},ylabel={Noise Level (dBA)}},
				point meta min=7,	point meta max=43,
				xmin=0, xmax=4000, ymin=250, ymax=3750,
				xtick={0,1000,2000,3000,4000}, ytick={0,1000,2000,3000,4000},
				xticklabels={-2,-1,0,1,2}, yticklabels={-2,-1,0,1,2},
				axis on top, tick style={line width=0.5pt, black},
			},
			every patch/.append style={
				shader=flat, draw=black, line width=0.3pt,
			}
		}

    \subfigure{%
  \resizebox{0.98\textwidth}{!}{%
    \noindent\begin{minipage}[c]{\textwidth}
      \begin{minipage}[c]{5em}
          \resizebox{1.0\linewidth}{!}{
            \makebox[6em][l]{(a) $\TrueNoise$ vs.\ $\Speed$}%
          }%
    \end{minipage}%
      \begin{minipage}[c]{\dimexpr\textwidth-5em\relax}
        \includegraphics[width=\linewidth]{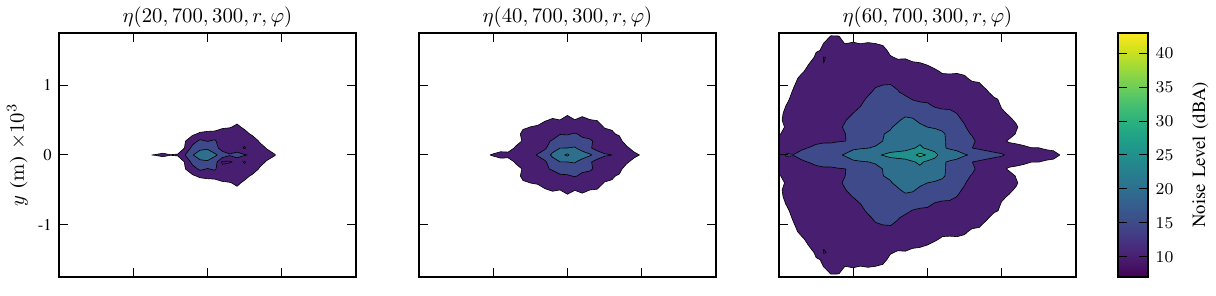}
      \end{minipage}%
    \end{minipage}%
  }%
  \label{fig:subfig_1}%
}%
\\[-0.8em]


    \subfigure{%
      \resizebox{0.98\textwidth}{!}{%
        \noindent\begin{minipage}[c]{\textwidth}
          \begin{minipage}[c]{5em}
              \resizebox{1.0\linewidth}{!}{
                \makebox[6em][l]{(b) $\TrueNoise$ vs.\ $\RPM$}%
              }%
        \end{minipage}%
          \begin{minipage}[c]{\dimexpr\textwidth-5em\relax}
            \includegraphics[width=\linewidth]{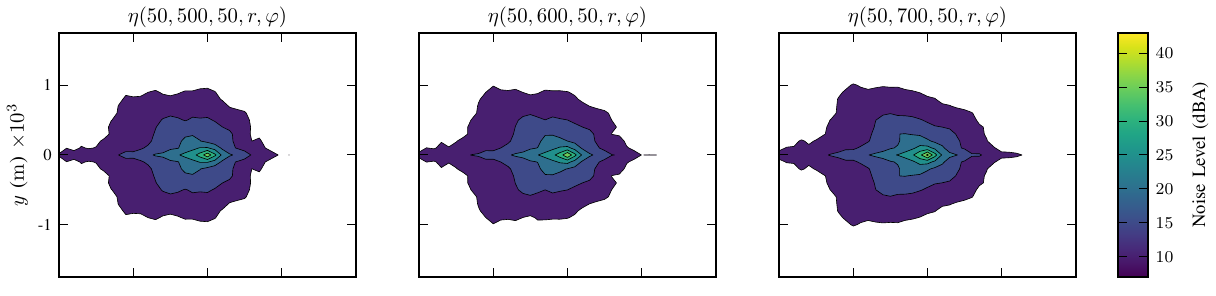}
          \end{minipage}%
        \end{minipage}%
      }%
      \label{fig:subfig_2}%
    }%
    \\[-0.8em]
            

    \subfigure{%
      \resizebox{0.98\textwidth}{!}{%
        \noindent\begin{minipage}[c]{\textwidth}
          \begin{minipage}[c]{5em}
              \resizebox{1.0\linewidth}{!}{
                \makebox[6em][l]{(c) $\TrueNoise$ vs.\ $\Vdistance$}%
              }%
        \end{minipage}%
          \begin{minipage}[c]{\dimexpr\textwidth-5em\relax}
            \includegraphics[width=\linewidth]{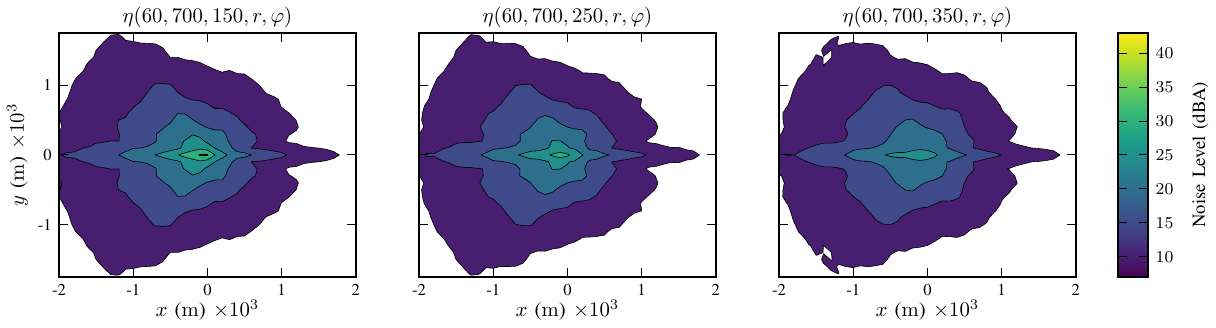}
          \end{minipage}%
        \end{minipage}%
      }%
      \label{fig:subfig_3}%
    }%

    \caption{Variation of noise contours with respect to
			(a) speed $\Speed = 20, 40, 60$~\si{\meter\per\second};
			(b) rotor RPM $\RPM = 500, 600, 700$~\si{\rpm}; and
			(c) altitude $\PosZ = 50, 150, 250$~\si{\meter}.
		The noise contours show a monotonic increase in the OASPL as the speed $\Speed$ and rotor RPM $\RPM$ increase, and a monotonic decrease as the altitude $\Vdistance$ increases. Contour levels are every \SI{5} dBA.}
    \label{fig:noise_contour}
\end{figure*} 

\begin{figure}[t]
	\centering
	\setlength{\FigHeight}{2.6in}
	\setlength{\FigWidth}{4.0in}
	\includegraphics[width=0.9\columnwidth]{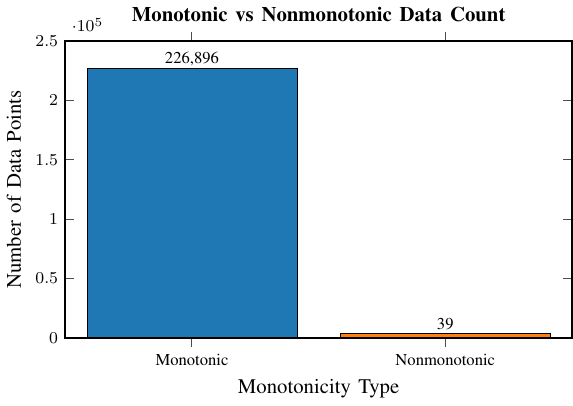}
	\caption{Distribution of monotonic vs non-monotonic data points in the dataset.
	The monotonicity assumption holds for $>\!\SI{99.98}{\percent}$ of the data points ($n = \SI{226935}{}$).}
	\label{fig:monotonicity_test}
\end{figure}

\experiment{Correctness of the Theoretical Upper Bound}
In this experiment, we validate the correctness of the theoretical upper bound on the noise prediction model derived in~\thmref{thm:error_bound}
by comparing the noise prediction error with the such bound.
The comparison is done using a validation dataset that contains noise samples that were not used during the training of the monotonic NN-based noise model, uniformly sampled over the entire range of the input variables.
The sector-specific certified bound $\ErrorBound_{\SectorIndex}$ is computed for each sector $\SectorIndex \in [1:22]$ as provided by~\thmref{thm:error_bound}.
\figref{fig:exp_true_vs_nn} visualizes the spatial distribution of the true noise levels,
the predicted values, and the corresponding prediction errors.
\figref{fig:formal_guarantee_unseen} compares those errors with the certified bounds of the corresponding sectors.
As shown, all prediction errors lie below their corresponding certified bounds,
thereby empirically validating the correctness of the certified bounds derived in~\thmref{thm:error_bound}.
This demonstrates that the sector-specific certified learning framework can accurately predict the noise levels of the eVTOL while providing formal guarantees on the prediction errors.

\begin{figure*}[t]\setlength{\subfigcapskip}{-0.5em}
    \centering
		\setlength{\FigHeight}{2.25in}
		\setlength{\FigWidth}{2.6in}
		\setlength{\FigXGap}{0.2in}
		\pgfplotsset{
			every axis/.append style={
				height=\FigHeight, width=\FigWidth, enlargelimits=false,
				title style={yshift=-0.7em, xshift=0em},
				xlabel={$\PosX$ (\si{\meter}) $\times 10^3$}, ylabel={$\PosY$ (\si{\meter}) $\times 10^3$},
				xlabel style={yshift=0.5em}, ylabel style={yshift=-0.2em},
				colormap/viridis,	colorbar=true,
				colorbar style={width=0.3cm, ylabel={Noise Level (dBA)}, ylabel style={yshift=0.3em}, at={(1.02,0.5)}, anchor=west},
				xmin=0, xmax=4000, ymin=250, ymax=3750,
				xtick={0,1000,2000,3000,4000}, ytick={0,1000,2000,3000,4000},
				xticklabels={-2,-1,0,1,2}, yticklabels={-2,-1,0,1,2},
				axis on top, tick style={line width=0.5pt, black},
			},
			every patch/.append style={
				shader=flat, draw=black, line width=0.3pt,
			}
		}

    \subfigure{%
      \resizebox{0.98\textwidth}{!}{%
        \noindent\begin{minipage}[c]{\textwidth}
          \begin{minipage}[c]{2em}
              \resizebox{1.0\linewidth}{!}{
                \makebox[2em][l]{(a)}%
              }%
        \end{minipage}%
          \begin{minipage}[c]{\dimexpr\textwidth-2em\relax}
            \includegraphics[width=\linewidth]{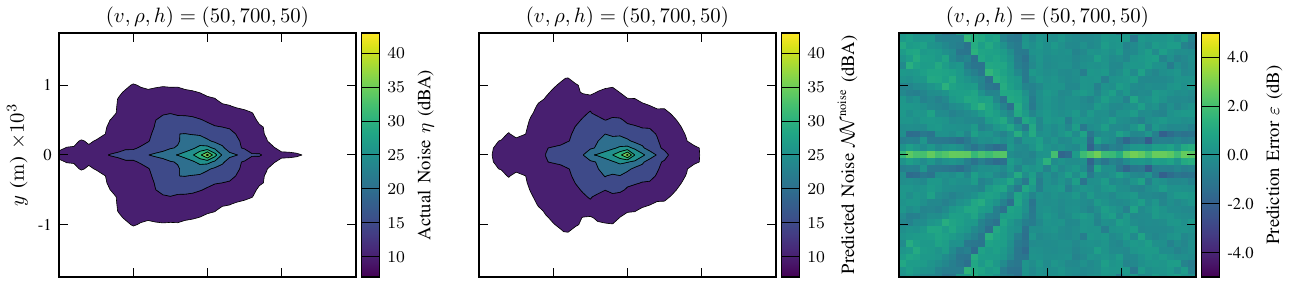}
          \end{minipage}%
        \end{minipage}%
      }%
      \label{fig:exp_true_vs_nn:a}%
    }%
    \\[-0.8em]


    \subfigure{%
      \resizebox{0.98\textwidth}{!}{%
        \noindent\begin{minipage}[c]{\textwidth}
          \begin{minipage}[c]{2em}
              \resizebox{1.0\linewidth}{!}{
                \makebox[2em][l]{(b)}%
              }%
        \end{minipage}%
          \begin{minipage}[c]{\dimexpr\textwidth-2em\relax}
            \includegraphics[width=\linewidth]{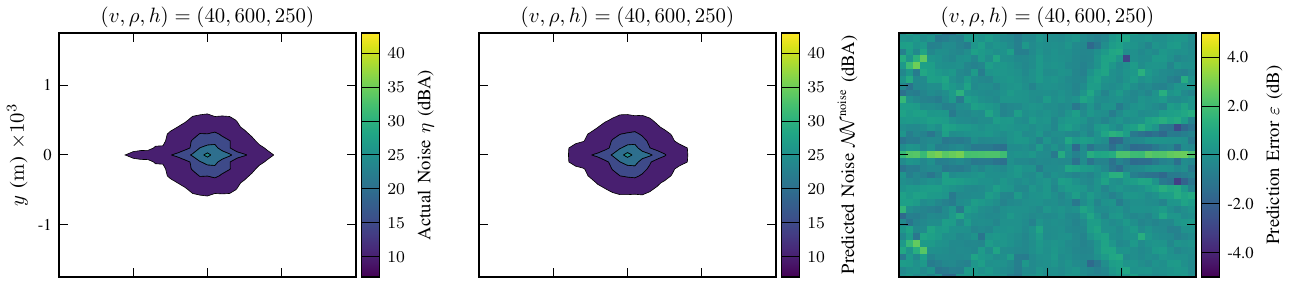}
          \end{minipage}%
        \end{minipage}%
      }%
      \label{fig:exp_true_vs_nn:b}%
    }%
    \\[-0.8em]
    
    \subfigure{
      \resizebox{0.98\textwidth}{!}{%
        \noindent\begin{minipage}[c]{\textwidth}
          \begin{minipage}[c]{2em}
              \resizebox{1.0\linewidth}{!}{
                \makebox[2em][l]{(c)}%
              }%
        \end{minipage}%
          \begin{minipage}[c]{\dimexpr\textwidth-2em\relax}
            \includegraphics[width=\linewidth]{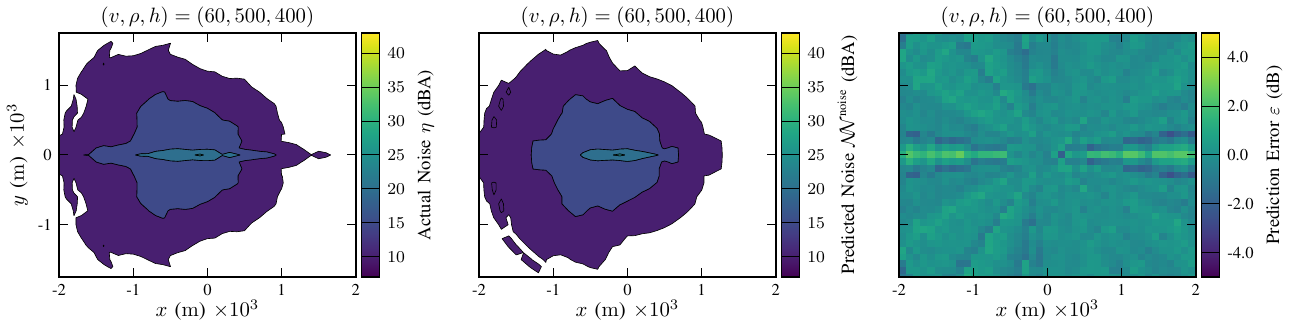}
          \end{minipage}%
        \end{minipage}%
      }%
      \label{fig:exp_true_vs_nn:c}%
    }
            
    \caption{Comparison of the actual noise levels $\TrueNoise(\Speed,\Rpm,\Vdistance,\Hdistance,\Azimuth)$ (left), the predicted noise levels $\NN^{\textup{noise}}(\Speed,\Rpm,\Vdistance,\Hdistance,\Azimuth)$ (middle), and the prediction errors $\Error(\Speed,\Rpm,\Vdistance,\Hdistance,\Azimuth)$ (right) for three different flight conditions. Contour levels are every \SI{5} dBA.}
    \label{fig:exp_true_vs_nn}
\end{figure*} 

\begin{figure}[t]
    \centering
		\setlength{\FigHeight}{2.1in}
		\setlength{\FigWidth}{5.0in}
    \includegraphics[width=0.98\columnwidth]{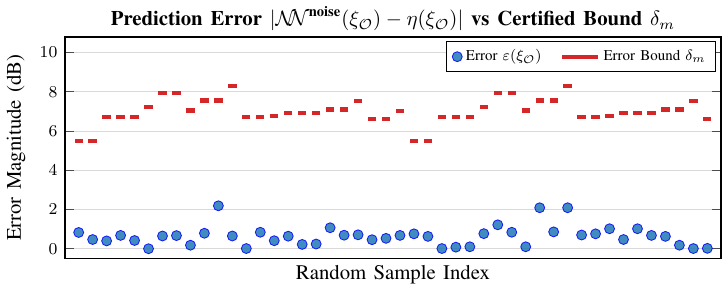}        
    \caption{Comparison of the noise prediction error (blue circles) with the certified bound for the corresponding sector (red dashes).}
    \label{fig:formal_guarantee_unseen}
\end{figure}

\subsection{Active Sampling for Noise Prediction}

\experiment{Evaluation of the Active Sampling Strategy}
In this experiment, we evaluate the effectiveness of the active sampling strategy in  tightly bounding the absolute error of the NN-based noise model collecting optimal amount of data.
Specifically, we compare the error bounds obtained using a uniform sampling strategy (baseline) with those obtained using the active sampling strategy described in~\secref{sec:active_sampling}.
For both strategies, \tabref{tab:active_sampling} reports the obtained sector-specific certified bounds $\ErrorBound_{\SectorIndex}$ defined in~\thmref{thm:error_bound} and their components $\Term_1$, $\Term_2$, and $\Term_3$.
The results show that the active sampling strategy significantly reduces the certified bounds compared to the uniform sampling strategy.
The relatively large error bounds provided by the baseline sampling strategy, ranging from $\SI{8.83}{dB}$ to $\SI{19.54}{dB}$,
highlight the need for a finer discretization of the input space to provide tighter bounds on the noise prediction errors.
However, collecting noise samples is computationally expensive, and making the grid uniformly finer inflates the sample count without guaranteeing a commensurate reduction in $\Term_1$. Under a baseline uniform strategy, extra samples are spread over the entire input space, so many fall in hypercubes that already satisfy the local-variation threshold and thus do not decrease $\Term_1$. By contrast, the proposed active sampling strategy refines only those hypercubes whose variation exceeds
$\BoundedVarAct$ and continues until $\Term_1 \leq \BoundedVarAct$. Consequently, to reach the same target—e.g., $\BoundedVarAct = \SI{1.5}{dB}$—uniform sampling must use substantially more samples overall than active sampling, which collects only the data necessary (near-minimally) to drive $\Term_1$ below the threshold. In our experiments, this targeted discretization yields much tighter certified bounds—$\SI{5.24}{dB}$ to $\SI{7.96}{dB}$—while keeping the sample count reasonable.

\begin{table}[t]
\centering
\caption{Sector-specific certified bounds for noise prediction under active sampling strategies.}
\label{tab:active_sampling}
\setlength{\tabcolsep}{1.0em}%
\resizebox{0.95\columnwidth}{!}{
\begin{tabular}{c|cccc|cccc}
\toprule
\multirow{2}{*}{$\SectorIndex$} & \multicolumn{4}{c|}{Uniform Sampling (Baseline)} & \multicolumn{4}{c}{Active Sampling ($\BoundedVarAct$ = 1.5)} \\[0.5ex]
 & $\Term_1$ & $\Term_2$ & $\Term_3$ & $\ErrorBound_\SectorIndex$ & $\Term_1$ & $\Term_2$ & $\Term_3$ & $\ErrorBound_\SectorIndex$ \\
\midrule
1  & \num{3.68} & \num{3.57} & \num{0.58} & \phantom{0}\num{8.83} & \num{1.50} & \num{1.93} & \num{0.81} & \num{5.24} \\
2  & \num{3.77} & \num{3.79} & \num{0.45} & \phantom{0}\num{9.01} & \num{1.50} & \num{1.91} & \num{1.08} & \num{5.49} \\
3  & \num{8.82} & \num{8.83} & \num{0.89} & \num{19.54} & \num{1.50} & \num{2.27} & \num{2.16} & \num{6.93} \\
4  & \num{4.02} & \num{3.71} & \num{0.97} & \phantom{0}\num{9.70} & \num{1.50} & \num{2.30} & \num{1.80} & \num{6.60} \\
5  & \num{4.15} & \num{3.79} & \num{0.95} & \phantom{0}\num{9.89} & \num{1.50} & \num{2.41} & \num{1.82} & \num{6.73} \\
6  & \num{5.20} & \num{4.93} & \num{2.08} & \num{13.21} & \num{1.50} & \num{2.09} & \num{2.14} & \num{6.73} \\
7  & \num{5.04} & \num{4.60} & \num{1.19} & \num{11.83} & \num{1.50} & \num{2.91} & \num{1.83} & \num{7.24} \\
8  & \num{6.72} & \num{6.68} & \num{2.18} & \num{16.58} & \num{1.50} & \num{3.15} & \num{2.31} & \num{7.96} \\
9  & \num{5.07} & \num{4.47} & \num{1.96} & \num{12.50} & \num{1.50} & \num{2.35} & \num{2.20} & \num{7.05} \\
10 & \num{5.21} & \num{4.87} & \num{2.74} & \num{13.82} & \num{1.50} & \num{2.16} & \num{2.90} & \num{7.56} \\
11 & \num{4.48} & \num{3.77} & \num{1.66} & \num{10.91} & \num{1.50} & \num{2.38} & \num{3.43} & \num{8.31} \\
12 & \num{4.01} & \num{3.96} & \num{0.41} & \phantom{0}\num{9.38} & \num{1.50} & \num{2.48} & \num{2.04} & \num{7.02} \\
13 & \num{4.06} & \num{4.01} & \num{0.59} & \phantom{0}\num{9.66} & \num{1.50} & \num{2.52} & \num{1.96} & \num{6.98} \\
14 & \num{4.15} & \num{3.79} & \num{0.96} & \phantom{0}\num{9.90} & \num{1.50} & \num{2.40} & \num{1.82} & \num{6.72} \\
15 & \num{3.60} & \num{3.35} & \num{1.63} & \phantom{0}\num{9.58} & \num{1.50} & \num{1.89} & \num{1.71} & \num{6.10} \\
16 & \num{3.94} & \num{3.25} & \num{2.12} & \num{10.31} & \num{1.50} & \num{1.99} & \num{2.26} & \num{6.75} \\
17 & \num{5.08} & \num{4.04} & \num{1.18} & \num{11.90} & \num{1.50} & \num{2.20} & \num{2.20} & \num{6.90} \\
18 & \num{8.60} & \num{8.45} & \num{2.14} & \num{20.19} & \num{1.50} & \num{1.90} & \num{2.50} & \num{6.90} \\
19 & \num{4.68} & \num{3.47} & \num{2.12} & \num{11.27} & \num{1.50} & \num{2.20} & \num{2.40} & \num{7.10} \\
20 & \num{4.93} & \num{4.53} & \num{2.23} & \num{12.69} & \num{1.50} & \num{2.00} & \num{3.04} & \num{7.54} \\
21 & \num{5.12} & \num{4.04} & \num{1.44} & \num{11.60} & \num{1.50} & \num{2.15} & \num{2.45} & \num{7.10} \\
22 & \num{3.89} & \num{3.34} & \num{1.76} & \phantom{0}\num{9.99} & \num{1.50} & \num{2.07} & \num{2.45} & \num{7.02} \\
\bottomrule
\end{tabular}}
\end{table}

\

\subsection{Noise-Aware Motion Planning for a Single eVTOL}
In the next set of experiments, we use the proposed noise-aware kinodynamic RRT* motion planning framework to plan a trajectory for an eVTOL navigating through an urban airspace while respecting noise ordinances defined at designated observer zones.
The simulation environment includes a 3D airspace bounded by
$(\PosX, \PosY, \PosZ) \in [0, 2200] \times [0, 2200] \times [0, 450] \SI{}{\,\meter}$.
It also includes three NAZs $\{\NAZ_j = (\Observer_j, \InsNoiseThreshold_j, \AvgNoiseThreshold_j, {\Window_j}, \Ordinance_j) \mid j \in \{1, 2, 3\}\}$, with each NAZ containing a single observer $\Observer_j$ located on the ground.
The eVTOL operates under a discrete-time kinematic model with a fixed time step of $\Window = \SI{5}{\second}$.
A pre-trained monotonic NN-based noise model $\NoiseModel^{\textup{noise}}$ is used by the framework to predict the noise levels at each observer based on the current eVTOL state.

The transition from current state $\State$ to the next state $\State'$ is determined by the kinematic model and the control input, where the latter is bounded by the following constraints:
\begin{align*}
|\Delta \Velocity| \leq \SI{5}{\meter\per\second^2},\quad
|\Delta \Vdistance | \leq \SI{5}{\meter\per\second},\quad
|\Delta \Heading| \leq \SI{5}{\degree\per\second}.
\end{align*}
In this scenario, we evaluate the behavior of a single-eVTOL system under the proposed framework.
All three NAZs share identical noise ordinances to set a unified upper bound on the allowable noise levels.

\experiment{Motion Planning for a Single eVTOL}
We consider three scenarios, defined by level of noise constraints specified in Definition~\ref{def:naz} as follows:
\begin{itemize}
	\item Relaxed constraints : $\InsNoiseThreshold = 45~\text{dBA}$, $\AvgNoiseThreshold = 43~\text{dBA}$. \\(relatively loud)
	\item Moderate constraints : $\InsNoiseThreshold = 35~\text{dBA}$, $\AvgNoiseThreshold = 30~\text{dBA}$. \\(quiet)
	\item Strict constraints : $\InsNoiseThreshold = 22~\text{dBA}$, $\AvgNoiseThreshold = 20~\text{dBA}$. \\(very quiet)
\end{itemize}
\figref{fig:exp_mp_single} shows the planned trajectories, noise profiles, and cruise speed for each of the three scenarios.
Under relaxed noise constraints (\figref{fig:exp_mp_single:a}), the planner can prioritize shorter trajectories by minimizing the distance traveled and maximizing the cruise speed, reaching the goal within \SI{60}{\second}.
Under moderate noise constraints (\figref{fig:exp_mp_single:b}), the planner adjusts the trajectory to avoid the central observer, resulting in a slightly longer path while maintaining a similar cruise speed, reaching the goal in \SI{65}{\second}.
Under strict noise constraints (\figref{fig:exp_mp_single:c}), the planner further modifies the trajectory to avoid all observers, increasing the average altitude and reducing the cruise speed, resulting in a longer path and a total travel time of \SI{90}{\second}.
This demonstrates the planner's ability to adapt to varying noise constraints while optimizing for travel time and distance.

\begin{figure*}[t] 
	\centering
	\pgfplotsset{
		every axis/.append style={
			grid=none,
			title style={yshift=-0.5ex},
			xlabel style={yshift=0.5ex},
			}
		}
	\newcommand{\AxisLabelSize}{\normalsize}
	\setlength{\FigWidth}{1.9in}\setlength{\FigHeight}{1.6in}
	\setlength{\FigXGap}{0.5in}\setlength{\FigYGap}{-0.1in}
	\setlength{\YlabelYShift}{-5pt}\setlength{\TitleYShift}{-10pt}
	\setlength{\LineWidth}{1.6pt}\setlength{\MarkSize}{2.6pt}

    \subfigure{
      \resizebox{0.98\textwidth}{!}{%
        \noindent\begin{minipage}[c]{\textwidth}
          \begin{minipage}[c]{2.5em}
              \resizebox{1.0\linewidth}{!}{
                \makebox[2em][l]{(a)}%
              }%
        \end{minipage}%
          \begin{minipage}[c]{\dimexpr\textwidth-3em\relax}
            \includegraphics[width=\linewidth]{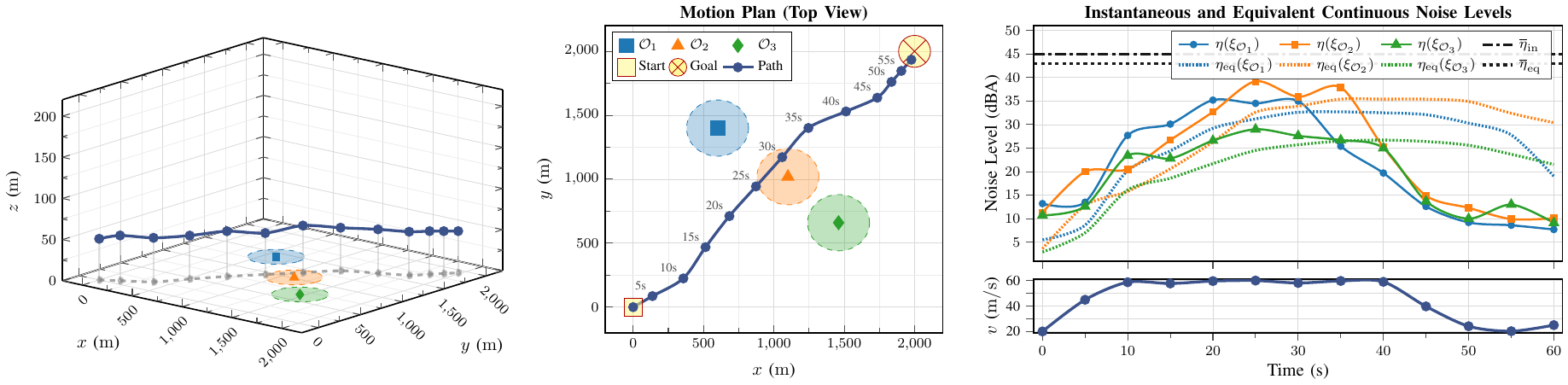}
          \end{minipage}%
        \end{minipage}%
      }%
      \label{fig:exp_mp_single:a}%
    }\vspace{-0.8em}
    
    \subfigure{
      \resizebox{0.98\textwidth}{!}{%
        \noindent\begin{minipage}[c]{\textwidth}
          \begin{minipage}[c]{2.5em}
              \resizebox{1.0\linewidth}{!}{
                \makebox[2em][l]{(b)}%
              }%
        \end{minipage}%
          \begin{minipage}[c]{\dimexpr\textwidth-3em\relax}
            \includegraphics[width=\linewidth]{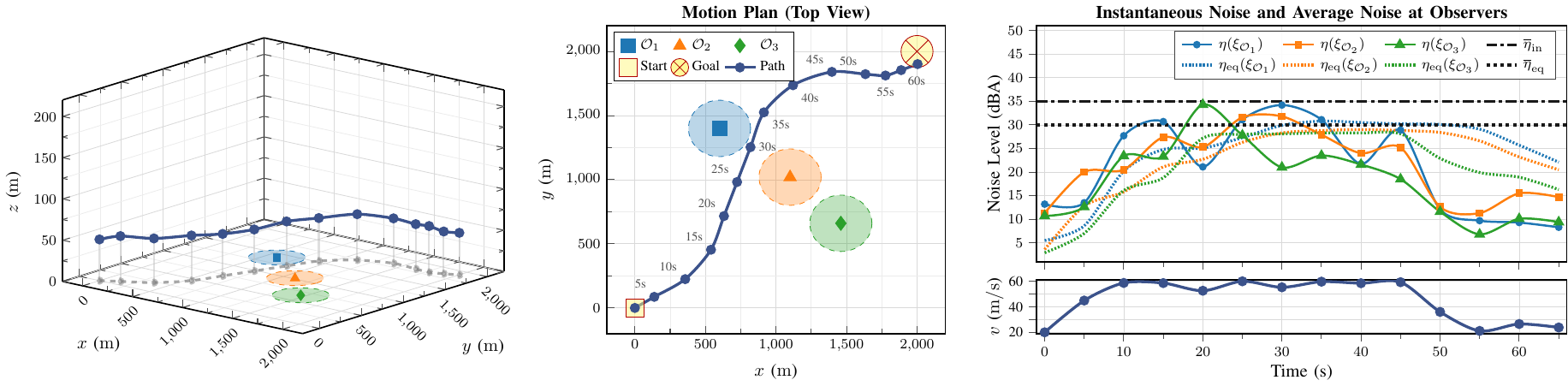}
          \end{minipage}%
        \end{minipage}%
      }%
      \label{fig:exp_mp_single:b}%
    }\vspace{-0.8em}

    \subfigure{
      \resizebox{0.98\textwidth}{!}{%
        \noindent\begin{minipage}[c]{\textwidth}
          \begin{minipage}[c]{2.5em}
              \resizebox{1.0\linewidth}{!}{
                \makebox[2em][l]{(c)}%
              }%
        \end{minipage}%
          \begin{minipage}[c]{\dimexpr\textwidth-3em\relax}
            \includegraphics[width=\linewidth]{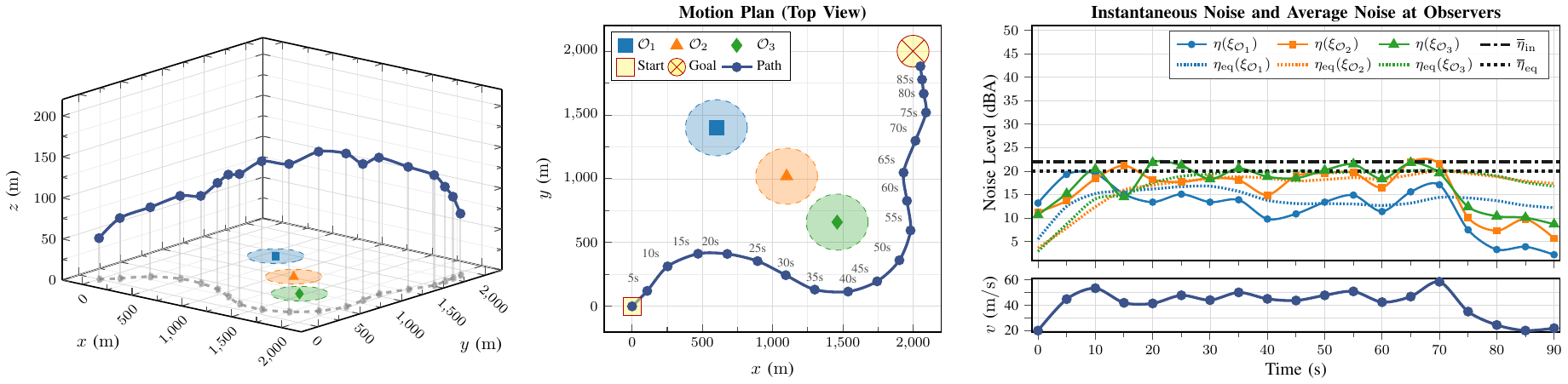}
          \end{minipage}%
        \end{minipage}%
      }%
      \label{fig:exp_mp_single:c}%
    }
	\caption{Planned trajectories and cruise speeds for the eVTOL in Experiment 1, and the corresponding noise levels at the observer locations.}
	\label{fig:exp_mp_single}
\end{figure*}

\subsection{Noise-Aware Motion Planning for Multiple eVTOLs}

\experiment{Motion Planning for Multiple eVTOLs}
This experiment evaluates the effectiveness of our noise-aware kinodynamic RRT* planner in coordinating multiple eVTOLs while ensuring compliance with dynamic noise constraints and avoiding collisions.
We consider a scenario involving three eVTOLs, each departing from a different vertiport and navigating to its own goal location.
The planning process follows an on-demand, sequential approach as described in \probref{prob:motion_planning},
on a first-come, first-served basis.
Once an eVTOL's motion plan is generated, it becomes fixed, constraining the planning space for subsequent eVTOLs.
	In this framework, the effective noise ordinance at each NAZ is dynamically updated to account for the noise contributions of previously scheduled eVTOLs. Formally, the remaining noise budget is given by:

\begin{align*}
  {\overline{\Noise}^*}_{\!\!\!\Instant}(\State_{n,\Observer}^{(t)}) &= 10 \log_{10}\!\!\left(10^{\,{\InsNoiseThreshold}/10}-\sum_{i=1}^{n}10^{\,\Noise(\State_{i,\Observer}^{(t)})/10}\right), \\
{\overline{\Noise}^*}_{\!\!\!\Average}(\State_{n,\Observer}^{(t)})
&= 10 \log_{10}\!\!\left(10^{{\AvgNoiseThreshold}/10}-
   \sum_{i=1}^{n} \frac{1}{\Window} \!\!\!\!\!\!\!\!\sum_{\;\;\;\;\;\;j = t - \Window}^{t} \!\!\!\!\!\!\!\!\!10^{\,\Noise(\State_{i,\Observer}^{(j)})/10}
  \right),
  \end{align*}
  where $\InsNoiseThreshold$ and $\AvgNoiseThreshold$ are equal to
$\NAZ^{\ErrorBound}\!\!. \InsNoiseThreshold$ and
$\NAZ^{\ErrorBound}\!\!. \AvgNoiseThreshold$, respectively; these are the noise limits from 
$\ErrorBound$-tightened NAZ.
  Let ${{\overline{\Noise}^*}_{\!\!\!\Instant}}(\State_{n,\Observer}^{(t)})$ and ${{\overline{\Noise}^*}_{\!\!\!\Average}}(\State_{n,\Observer}^{(t)})$ denote the renewed instantaneous and equivalent-continuous noise thresholds for the $n$-th eVTOL to be planned, that is, after the trajectories of the first $n-1$ eVTOLs have been fixed. We define $\State_{i,\Observer}^{(t)}$ be the state of eVTOL $i$ relative to observer location $\Observer$
at time $t$ (this extends~\eqref{eq:relative_state} to multiple vehicles). Planning proceeds sequentially: once $\Plan_1,\Plan_2,\dots,\Plan_{n-1}$ are fixed, renew the budgets using the formulas above, then plan $\Plan_{n}$ subject to the updated thresholds ${{\overline{\Noise}^*}_{\!\!\!\Instant}}(\State_{n,\Observer}^{(t)})$ and ${\overline{\Noise}^*}_{\!\!\!\Average}(\State_{n,\Observer}^{(t)})$.

Spatial separation is enforced through a spherical safety boundary of radius \SI{100}{\meter} surrounding each eVTOL at every timestep.
During RRT* tree expansion, nodes violating this boundary relative to any previously scheduled eVTOL are discarded.
The noise and collision constraints are time-varying due to the presence of other agents.

\figref{fig:exp_mp_multi} shows the motion plans for all three eVTOLs.
We observe that later eVTOLs tend to take longer and more elevated paths as the feasible planning space progressively shrinks due to the noise contributions of earlier flights.
We also note that tree pruning becomes more aggressive for later-scheduled eVTOLs, as the noise constraints become more restrictive, though the planner maintains reliable convergence towards a solution.

\begin{figure*}[t] 
	\centering
	\pgfplotsset{
		every axis/.append style={
			grid=none,
			title style={yshift=-0.5ex},
			xlabel style={yshift=0.5ex},
			}
		}
	\setlength{\FigWidth}{1.9in}\setlength{\FigHeight}{1.6in}
	\setlength{\FigXGap}{0.5in}\setlength{\FigYGap}{-0.1in}
	\setlength{\YlabelYShift}{-5pt}\setlength{\TitleYShift}{-10pt}
	\setlength{\LineWidth}{1.6pt}\setlength{\MarkSize}{2.6pt}
	%
    \includegraphics[width=0.95\linewidth]{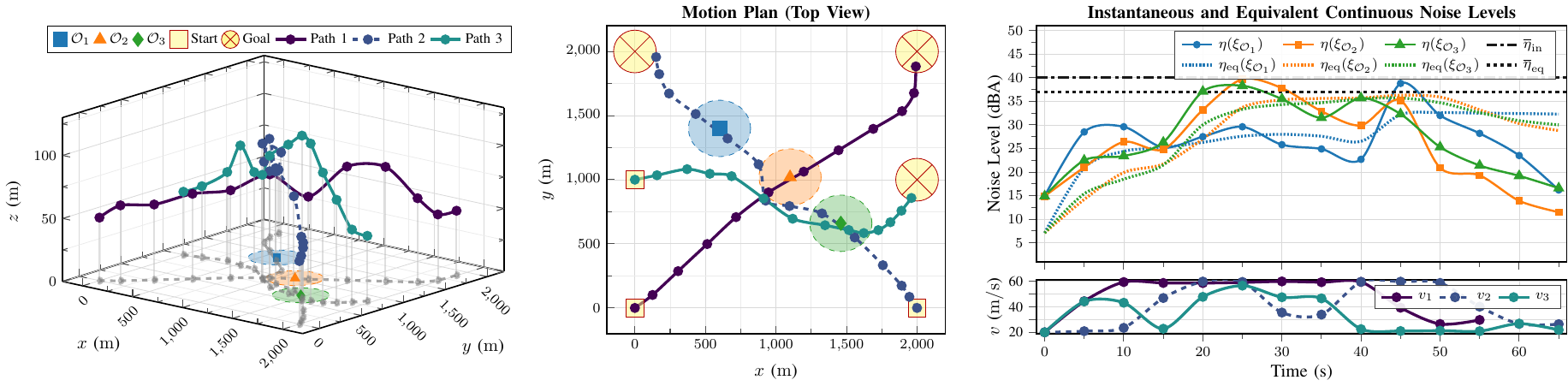}  
	\caption{Planned trajectories and cruise speeds for multiple eVTOLs, and the corresponding noise levels at the observer locations.}
	\label{fig:exp_mp_multi}
\end{figure*}

\subsection{Physics-Based Control Input Sampling}

\experiment{URS vs PBS Strategies}
In this experiment, we evaluate the impact of different control input sampling strategies on the performance of the proposed motion planning framework.
Specifically, we compare the two sampling strategies described earlier in~\secref{sec:motion_planning:pbs}:
\begin{itemize}
  \item URS, where control inputs $(\Speed, \Vdistance)$ are uniformly sampled from the entire admissible set; and 
  \item PBS, where $(\Speed, \Vdistance)$ are uniformly sampled from the input domain after pruning all tuples with higher $\Speed$ and lower $\Vdistance$ (see~\algref{alg:steer_pbs}).
\end{itemize}
The same environment setup as in~\secref{sec:experiments:noise_correctness} is used.  
For each steering strategy we execute 10 independent trials, recording the total planner iterations required to reach the goal.
\figref{fig:exp_pbs} shows the results of the experiment under two different levels of noise constraints.
Under relaxed noise constraints (\figref{fig:exp_pbs:a}), both sampling strategies exhibit similar performance.
Under tightened noise constraints (\figref{fig:exp_pbs:b}), the PBS strategy consistently outperforms the URS, requiring fewer number of iterations to reach the goal as the former systematically prunes high-noise regions from the search space.
These findings confirm that exploiting monotonicity property of the noise function can accelerate kinodynamic RRT* planning when noise ordinances are tight.

\begin{figure}[t]
  \centering
	\setlength{\FigWidth}{2.5in}\setlength{\FigHeight}{2.6in}
	\setlength{\FigXGap}{0.5in}\setlength{\FigYGap}{-0.1in}
	\pgfplotsset{
		BoxplotStyle/.style={
			fill=C7!20!white,
			boxplot/every median/.style={line width=1.5pt, draw=C3},
			boxplot/box extend=0.48,
		}
	}
  \subfigure[$\InsNoiseThreshold \! = \! 40~\text{dBA}$, $\AvgNoiseThreshold \! = \! 35~\text{dBA}$]
    {\includegraphics[width=0.48\columnwidth]{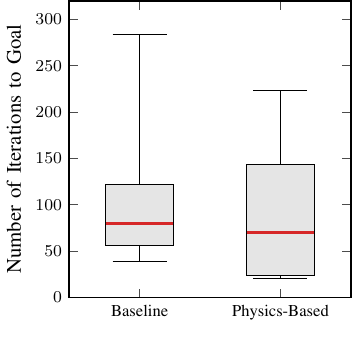}\label{fig:exp_pbs:a}}\hfill
  \subfigure[$\InsNoiseThreshold \! = \! 30~\text{dBA}$, $\AvgNoiseThreshold \! = \! 25~\text{dBA}$]
    {\includegraphics[width=0.48\columnwidth]{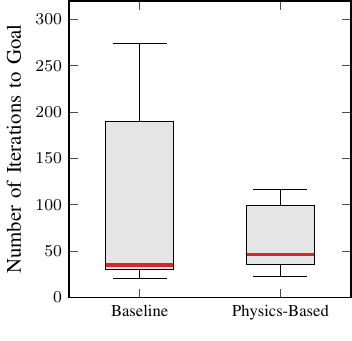}\label{fig:exp_pbs:b}}
  \caption{Comparison of iteration counts using Uniform Random Sampling (URS) and Physics-Based Sampling (PBS) under (a) relaxed and (b) strict noise constraints ($n=10$).
	}
  \label{fig:exp_pbs}
\end{figure}

\section{Conclusion}
\label{sec:conclusion}

In this paper, we presented a novel framework for noise-aware motion planning in eVTOLs.
The framework leverages a physics-based simulation environment to generate a comprehensive dataset using active sampling method for training a monotonic NN-based noise model, with formal guarantees on the bounds of the noise prediction errors.
This model is subsequently integrated into a kinodynamic RRT* planner, enabling the generation of noise-aware motion plans that adhere to community noise ordinances.
Our experiments demonstrated the effectiveness of the proposed approach in accurately predicting noise levels and improving planning efficiency through physics-based control input sampling strategies.
The results confirm that exploiting the monotonicity property of the noise function can accelerate kinodynamic RRT* planning when noise ordinances are tight. Furthermore, our findings highlight the importance of incorporating noise awareness into the planning process, paving the way for more efficient and reliable eVTOL operations in urban environments.

\bibliographystyle{IEEEtran}

\appendices
\section{Proof of \thmref{thm:error_bound}}\label{app:proofs}

Before we prove \thmref{thm:error_bound}, we first establish the following propositions.

\begin{proposition}
	\label{prop:norm_max}
	Let $\CompactSet \subset \Reals^n$ be a compact set and $\{\Hypercube^{i}\}_{i \in \IndexSet}$ be a hypercube discretization of $\CompactSet$ such that $\cup_{i\in\IndexSet} \Hypercube^{i} = \CompactSet$.
	Then, for any monotonic function $\FnFoo \colon \CompactSet \rightarrow \Reals$:
	\begin{align*}
		\max_{x \in \Reals^{n}} \big\|\FnFoo(x) \big\|
		\leq
	{\max_{x\in\Vertices(\Hypercube^{i})}} \big\|\FnFoo(x) \big\| .
	\end{align*}
    
\end{proposition}

\begin{proof}
	Assume for the sake of contradiction that there exists a point $x^* \in \Reals^{n}$ such that:
	\begin{align*}
		\big\|\FnFoo(x^*) \big\| >
		\max_{x \in \Vertices(\Hypercube)} \big\|\FnFoo(x) \big\| .
	\end{align*}
	Since
	$\cup_{i\in\IndexSet} \Hypercube^{i} = \CompactSet$ by definition,
	there exists a hypercube $\Hypercube^{*}$ such that $x^* \in \Hypercube^{*}$, hence:
	\begin{align*}
		\big\|\FnFoo(x^*) \big\| >
		\max_{x \in \Vertices(\Hypercube^{*})} \big\|\FnFoo(x) \big\| .
	\end{align*}
\begin{enumerate}
	\item If $x^* \in \Vertices(\Hypercube^{*})$, then $\text{LHS} = \text{RHS}$, which contradicts the assumption.
	\item If $x^* \notin \Vertices(\Hypercube^{*})$, then by the monotonicity of $\FnFoo$, we can find a point $x' \in \Vertices(\Hypercube^{*})$ such that:
	\begin{align*}
		\big\|\FnFoo(x') \big\| \geq \big\|\FnFoo(x^*) \big\|,
	\end{align*}
	which contradicts the assumption.
\end{enumerate}
Therefore, the proposition holds.
\end{proof}

\begin{proposition}
	\label{prop:const_monotonicity}
	Given a function $\FnFoo \colon D \rightarrow \Reals$ that is monotonic over a domain $D \subset \Reals^n$, and a constant $C \in \Reals$.
    Let $\FnBar \colon x \mapsto \FnFoo(x) - C$ be a function defined over the same domain of $D$. Then:
    \begin{align*}
    &\FnFoo(x')\ge \FnFoo(x) \implies \FnBar(x')\ge \FnBar(x),
    \;\text{s.t }\forall x'\ge x
    \end{align*}
\end{proposition}

\begin{proof}
	Without loss of generality, assume that $\FnFoo$ is monotonically non-decreasing with respect to $x_i$.
	Let $x,x' \in \Reals^n$ such that $x_i' \geq x_i$ and $x_j = x_j'$ for all $j \neq i$.
	Then, by the monotonicity of $\FnFoo$, we have:
	\begin{align*}
		\FnFoo(x') &\geq \FnFoo(x) \\
		\FnFoo(x') - C &\geq \FnFoo(x) - C &&\text{(since $C$ is constant)}\\
		\FnBar(x') &\geq \FnBar(x). &&\text{(by definition of $\FnBar$)}
	\end{align*}
	Thus, $\FnBar$ is also monotonically non-decreasing with respect to $x_i$.
	The same argument applies to any other variable in the input space, thus proving the proposition.
\end{proof}


\begin{proposition}
	\label{prop:argmax_monotonicity}
	Let $\FnFoo \colon D \rightarrow \Reals$ be a monotonic function defined over a domain
	$ D \subset \Reals^n $, and let $C \in \Reals$ be a constant.
	Then:
	\begin{align*}
		\argmax_{x \in D} \FnFoo(x) = \argmax_{x \in D} \left(\FnFoo(x) - C\right).
	\end{align*}
\end{proposition}

\begin{proof}
	Let $x^* = \argmax_{x \in D} \FnFoo(x)$,
	and let $\FnBar(x) = \FnFoo(x) - C$,
	and $x' = \argmax_{x \in D} \FnBar(x)$.
	Then, by definition:
	\begin{align*}
		\FnFoo(x^*) &= \max_{x \in D} \FnFoo(x) \\
		\FnFoo(x^*) - C &= \max_{x \in D} \left(\FnFoo(x) - C\right) \\
		\FnBar(x^*) &= \max_{x \in D} \FnBar(x) \\
		x^* &= \argmax_{x \in D} \FnBar(x) = x'.
	\end{align*}
	That is:
	\begin{align*}
		x^* = \argmax_{x \in D} \FnFoo(x) \implies x^* = \argmax_{x \in D} \FnBar(x).
	\end{align*}
	The reverse direction follows similarly, thus proving the proposition.
\end{proof}

Now, we can provide the proof of~\thmref{thm:error_bound} as follows.
\begin{proof}[Proof of~\thmref{thm:error_bound}]
The maximum prediction error of the neural network $\NN^{\textup{noise}}$ over the input domain $\InputDomain_{\Observer}$ is given by:
\begin{align*}
	\Error
	& = \;\; \max_{\State_{\Observer} \in \InputDomain_\Observer} \;\;
		\big| \TrueNoise(\State_{\Observer}) - \NN^{\textup{noise}} (\State_{\Observer})\big|\nonumber \\
	& =
		\max_{\substack{\SectorIndex \in [1:\SectorCt] \\ \State_{\Observer} \in \InputDomain_{\Observer,\SectorIndex}}}
		\big|\TrueNoise(\State_{\Observer}) - \NN_{\!\SectorIndex}^{\textup{noise}}(\State_{\Observer})\big|\nonumber \\
	& = {\max_{\substack{\SectorIndex \in [1:\SectorCt] \\ \State_{\Observer} \in \InputDomain_{\Observer,\SectorIndex}}}}
		\big|\TrueNoise(\State_{\Observer}) \!-\! \ConstC_{\SectorIndex}
	- \NN_{\!\SectorIndex}^{\textup{noise}}(\State_{\Observer}) \!+\! \ConstD_{\SectorIndex}
	+ \ConstC_{\SectorIndex} \!-\! \ConstD_{\SectorIndex}\big| .
\end{align*}
By the triangle inequality, we split the error into four terms:
\begin{align}
	\Error
	& \leq {\max_{\substack{\SectorIndex\in[1:\SectorCt]\\\State_{\Observer}\in\InputDomain_{\Observer,\SectorIndex}}}}
		\left\{\begin{aligned}
			&\big| \TrueNoise(\State_{\Observer}) - \ConstC_{\SectorIndex} \big|
		\nonumber \\
	& 
	+ \big| \NN_{\!\SectorIndex}^{\textup{noise}}(\State_{\Observer}) - \ConstD_{\SectorIndex} \big| 
	\nonumber \\
	& 
	+ \big| \ConstC_{\SectorIndex} - \ConstD_{\SectorIndex} \big|
	\end{aligned} \right\} \nonumber \\
	& = {\max_{\substack{\SectorIndex\in[1:\SectorCt]\\\State_{\Observer}\in\InputDomain_{\Observer,\SectorIndex}}}}
	\left\{ \begin{aligned}
		& \big| \TrueNoise(\State_{\Observer}) -\TrueNoise(\State_{\Observer}\mid\tilde{\Azimuth}_{\SectorIndex}) \nonumber \\
		&
		\phantom{\big| \TrueNoise(\State_{\Observer})}
		+ \TrueNoise(\State_{\Observer}\mid\tilde{\Azimuth}_{\SectorIndex}) - \ConstC_{\SectorIndex} \big| \nonumber \\
		& \;
		+ \big| \NN_{\!\SectorIndex}^{\textup{noise}}(\State_{\Observer}) - \ConstD_{\SectorIndex} \big| 
	+ \big| \ConstC_{\SectorIndex} - \ConstD_{\SectorIndex} \big|
	\end{aligned} \right\} \nonumber \\
	& \leq {\max_{\substack{\SectorIndex\in[1:\SectorCt]\\\State_{\Observer}\in\InputDomain_{\Observer,\SectorIndex}}}}
	\left\{ \begin{aligned}
		& \big| \TrueNoise(\State_{\Observer}) -\TrueNoise(\State_{\Observer}\mid\tilde{\Azimuth}_{\SectorIndex})\big| 
		\nonumber \\
		& \phantom{1}
		+	\big|\TrueNoise(\State_{\Observer}\mid\tilde{\Azimuth}_{\SectorIndex}) - \ConstC_{\SectorIndex} \big| \nonumber \\
		& \phantom{1}
		+ \big| \NN_{\!\SectorIndex}^{\textup{noise}}(\State_{\Observer}) - \ConstD_{\SectorIndex} \big|
		+ \big| \ConstC_{\SectorIndex} - \ConstD_{\SectorIndex} \big|
	\end{aligned} \right\}. \nonumber
\end{align}
From~\assumpref{assump:angular_monotonicity_noise},
we can bound the first term by the maximum noise variation within $\InputDomain_{\Observer,\SectorIndex}$, that is:
\begin{align}
	\Error
	& \leq {\max_{\substack{\SectorIndex\in[1:\SectorCt]\\\State_{\Observer}\in\InputDomain_{\Observer,\SectorIndex}}}}
	\left\{ \begin{aligned}
		& 1 +\big|\TrueNoise(\State_{\Observer}\mid\tilde{\Azimuth}_{\SectorIndex}) - \ConstC_{\SectorIndex} \big| \nonumber \\
		& \phantom{1}
		+ \big| \NN_{\!\SectorIndex}^{\textup{noise}}(\State_{\Observer}) - \ConstD_{\SectorIndex} \big| 
		+ \big| \ConstC_{\SectorIndex} - \ConstD_{\SectorIndex} \big|\nonumber
		\end{aligned} \right\} \nonumber .
\end{align}
Since both $\TrueNoise(\State_{\Observer}\mid\tilde{\Azimuth}_{\SectorIndex})$ and $\NN_{\SectorIndex}^{\textup{noise}}(\State_{\Observer})$ are monotonic by the \assumpref{assump:noise_monotonicity} and \defref{def:monotonic_neural_network}, we can apply~\propref{prop:const_monotonicity} and~\propref{prop:argmax_monotonicity} to conclude that the second and third terms are also monotonic with respect to the input variables.
Since the sector-specific input domain $\InputDomain_{\Observer,\SectorIndex}$ is compact, we can apply~\propref{prop:norm_max} to all four terms in the error bound.
Thus, the inequality holds for the maximum over the vertices of the hypercube $\Hypercube_{\SectorIndex}^{i j k \ell}$, that is:
\begin{align}
	\Error
	& \leq \hspace{-1.0em}
	{\max_{\substack{\SectorIndex\in[1:\SectorCt]\\(i,j,k,\ell)\in\IndexSet_{\SectorIndex}\\\State_{\Observer}\in\Vertices(\Hypercube_{\SectorIndex}^{i j k \ell})}}} \hspace{-1.2em}
	\left\{ \begin{aligned}
		& 1 +\big|\TrueNoise(\State_{\Observer}\mid\tilde{\Azimuth}_{\SectorIndex}) - \ConstC_{\SectorIndex} \big| \nonumber \\
		& \phantom{1}
		+ \big| \NN_{\SectorIndex}^{\textup{noise}}(\State_{\Observer}) - \ConstD_{\SectorIndex} \big| 
		+ \big| \ConstC_{\SectorIndex} - \ConstD_{\SectorIndex} \big|\nonumber
		\end{aligned} \right\} \nonumber .
\end{align}
From the definition of $\Dataset_{\Test}$ in~\eqref{eq:test_dataset}:
\begin{align}
	\Error
		& \leq {\max_{\substack{\SectorIndex\in[1:\SectorCt]\\\State_{\Observer} \in \Dataset_{\Test,\SectorIndex}}}}
	\left\{ \begin{aligned}
	& 1 + \big|\TrueNoise(\State_{\Observer}\mid\tilde{\Azimuth}_{\SectorIndex}) - \ConstC_{\SectorIndex} \big|  \nonumber \\
	& \phantom{1}
	+ \big| \NN_{\SectorIndex}^{\textup{noise}}(\State_{\Observer}) - \ConstD_{\SectorIndex} \big| 
	+ \big| \ConstC_{\SectorIndex} - \ConstD_{\SectorIndex} \big|
	\end{aligned} \right\}, \nonumber
\end{align}
which concludes the proof.
\end{proof}


\end{document}